\documentclass[lettersize,journal]{IEEEtran}
\usepackage{amsmath,amssymb,amsfonts}
\usepackage{algorithmic}
\usepackage{graphicx}
\usepackage{textcomp}
\usepackage{stfloats}
\usepackage{xcolor}
\usepackage{subfigure}
\usepackage{amsthm}
\newtheorem{theorem}{Theorem}
\usepackage{comment}
\usepackage{enumitem}
\usepackage{amsmath,amsfonts}

\newtheorem{remark}{Remark}

\usepackage{cite}
\begin{document}

\title{Over-the-Air Interference Nulling Using Active RIS}
\author{
\IEEEauthorblockN{Junzhi Wang, Jun Sun, Limin Liao, Xiangbai Liao, Yingzhuang Liu}\\
     \thanks{Junzhi Wang, Jun Sun, Limin Liao, and Yingzhuang Liu are with the School of Electronic Information and Communications, Huazhong University of Science and Technology, Wuhan, China.

     Xiangbai Liao is with the Hunan Institute of Technology, Hengyang, China.}

} 

\maketitle

\begin{abstract}
Interference fundamentally limits the performance of dense wireless networks, and reconfigurable intelligent surfaces (RIS) have recently emerged as a promising means of enabling interference-free transmission in the Degrees-of-Freedom (DoF) sense. This paper investigates the feasibility of achieving full DoF in a two-way 
$K$-user interference channel—a canonical interference-limited setting—by employing an active RIS. Unlike its passive counterpart, an active RIS is subject to both per-element gain constraints and a total reflection-power constraint, which renders over-the-air interference nulling equivalent to solving a constrained random linear system with coupled nonlinear constraints. By leveraging tools from high-dimensional convex geometry, we derive a tight scaling threshold on the required number of reflecting elements (REs) for full-DoF transmission. We further extend the analysis to scenarios where each RE incurs circuit power consumption under a total power budget, leading to a fundamental tradeoff between RIS transmit power and circuit power. For this setting, we establish the thresholds for both the total power and the corresponding number of REs required to achieve interference-free transmission. Simulation results validate the theoretical analysis.

\end{abstract}


\section{Introduction}
As 6G networks aim to support unprecedented connection densities, interference has become a fundamental bottleneck. Reconfigurable intelligent surfaces (RIS) have recently emerged as a promising tool for interference management \cite{bafghi2022degrees,jiang2022interference,chae2022cooperative,10706811}. However, passive RIS suffers from the \textit{multiplicative fading effect}: the path loss of the cascaded transmitter–RIS–receiver link equals the product of the individual reflective path losses, making it several orders of magnitude weaker than the direct link \cite{9306896}. As a result, thousands of passive reflecting elements (REs) are often required to overcome this severe attenuation, which challenges practical deployment.

In contrast, benefiting from amplification circuits, active RIS can not only reflect but also amplify the incident signals, thereby effectively mitigating multiplicative fading \cite{you2021wireless, yang2025robust, 9963962, 10319318}. 
However, to employ active RIS for interference nulling, the following key issues remain open and pose fundamental challenges for system design and analysis:
\begin{enumerate}[label=\arabic*)]  
\item Can active RIS achieve interference-free transmission under various hardware constraints?
\item How do these hardware constraints affect the required number of active REs?
\end{enumerate}

Mathematically, achieving interference-free transmission is equivalent to solving a constrained random linear system ${\bf{Ax}} + {\bf{b}} = 0$, where ${\bf{A}}$ and ${\bf{b}}$ denote the interference channels of the cascaded and direct links, and ${\bf{x}} \in \mathbb{C}^N$ represents the RIS adjustment. Without constraints, the system is feasible if ${\text{rank}}\left( {\mathbf{A}} \right) = {\text{rank}}\left( {\left[ {{\mathbf{A}}|{\mathbf{b}}} \right]} \right) \leqslant N$. In practice, hardware-imposed constraints on ${\bf{x}}$ make the feasibility analysis much more challenging. The main difficulties are summarized as follows:
\begin{enumerate}[label=\arabic*)]
\item 
\textit{Amplitude and Power constraints on RIS:}
For an active RIS, the amplitude (or amplification gain) of its element is normally bounded in practical implementations. Moreover, the total transmit power of an RIS is also bounded. These two constraints make the feasibility analysis of interference-free transmission more challenging.

\item 
\textit{Reflection–Circuit power tradeoff of RIS:} 
If the circuit power of each RE is also taken into account, there might be a tension between the reflection power and the circuit power consumption when the total power supply at the RIS is limited. Thus, it is essential to characterize both the minimum RIS power budget and the minimum number of REs required to ensure interference-free transmission.

\end{enumerate}

This paper addresses these challenges by geometrically formulating this feasibility problem as a set-intersection problem and employing tools from high-dimensional probability. 
The derived necessary and sufficient conditions not only demonstrate the feasibility of interference-free transmission but also reveal how various system parameters impact the required number of RE, thereby providing theoretical insights into the practical design.

\subsection{Prior Works}
Recently, RIS has attracted widespread attention in both academia and industry for its cost- and energy-efficient control of electromagnetic wave propagation. By appropriately adjusting the phase shifts, RIS can enhance desired signals and suppress interference at receivers \cite{peng2024two,shi2025combating,wang2021joint}. Compared with conventional interference alignment, which requires symbol extension and achieves a degree of freedom (DoF) of $\frac{K}{2}$ in the $K$-user interference channel \cite{cadambe2008interference}, recent studies have shown that RIS can completely eliminate interference without symbol extension and achieve the full DoF $K$ \cite{bafghi2022degrees,jiang2022interference,chae2022cooperative}. 
Specifically, \cite{jiang2022interference} showed that slightly more than $2K(K-1)$ passive REs are sufficient to achieve the full DoF of $K$ under weak direct channels, whereas stronger direct channels require more REs.

The reason behind the above phenomenon—that stronger direct channels require more REs—is twofold. First, stronger direct channels lead to stronger interference, thereby requiring more REs for mitigation. Second, due to the \textit{multiplicative fading effect}, the cascaded reflective link is typically several orders of magnitude weaker than the direct link \cite{9306896}. To compensate for this disparity, particularly under stronger direct channels, thousands of passive REs are typically required \cite{9998527}.

To overcome the \textit{multiplicative fading effect}, an enhanced RIS architecture, known as active RIS, has been proposed. By incorporating amplification circuits, active RIS can both reflect and amplify signals \cite{9998527, ZhouGui}. This capability helps compensate for the severe attenuation in the cascaded reflective links, thereby effectively mitigating the multiplicative fading and enhancing the received signal strength \cite{you2021wireless, yang2025robust, 9963962, 10319318}. It was shown in \cite{10319318} that, in a dual-function radar-communication system, active RIS can enhance the radar signal-to-interference-plus-noise ratio (SINR) by up to 70 dB compared to passive RIS.

By amplifying incident signals, active RIS can achieve the desired performance with fewer REs. Specifically, \cite{bafghi2022degrees} showed that the full DoF $K$ of the $K$-user interference channel can be achieved if the number of active REs exceeds $K(K-1)$. However, this analysis ignores amplification amplitude and power budget constraints. Considering these practical limits, \cite{10706811} showed that a necessary condition for active RIS to null interference is that the number of REs must be at least the rank of the cross-interference channels. Specifically, if the cross-interference channels are linearly independent, the required number of REs $N$ must satisfy $N \ge K(K-1)$. Furthermore, \cite{10706811} also demonstrated that when the cross-interference channels are extremely strong, interference-free transmission may become infeasible under a limited power budget.

Additionally, \cite{long2021active} showed that increasing the number of active REs raises circuit power consumption, which in turn reduces the power available for reflection and may render interference-free transmission infeasible.
However, the feasibility conditions, as well as the impact of the power budget, amplification constraints, and channel strength disparity on the required number of REs, remain unclear. 

\subsection{Main Contributions}
This paper aims to determine the feasibility conditions for interference-free transmission (i.e., full DoF) and characterize the impact of system parameters on the required number of active REs. Mathematically, the interference-free transmission problem is equivalent to analyzing the feasibility of a random linear system under joint norm and box constraints.
To derive a necessary condition, we first decompose the problem into two set-intersection subproblems and solve them using Gordon’s Theorem and the geometric properties of the subsets. Their results are then combined to yield the necessary condition for the original problem.
To derive a sufficient condition, we first approximate the intersection of the norm and box constraints with a smaller but more tractable set. Then, by exploiting the geometric relationship between this set and the interference-nulling equations, we establish a sufficient condition for the original problem. Furthermore, to ensure interference-free transmission under a limited total power budget, we also identify the feasibility conditions on the total power and the number of active REs.

Specifically, in a two-way $K$-user interference channel, let $\eta$ be the strength ratio between the direct and cascaded reflective links, and $\alpha$ the RE amplification factor. The reflection power at the RIS and the transmit power of each user are denoted by $P_r$ and $P$. Moreover, $\sigma_r^2$ and $\tau^2$ represent the RE thermal noise power and the user–RIS channel attenuation, respectively. Based on these definitions, we obtain the following key results:

\begin{enumerate}[label=\arabic*)]
\item 
\textbf{Without total power constraints:} 
If the RIS power supply is unconstrained, the necessary and sufficient conditions for achieving full DoF define a threshold number of REs, ${N_{\text{th}}}$, which scales on the same order\footnote{For the necessary condition, $\eta$ and $\tau^2$ can take values from the weakest direct and strongest reflective links. Conversely, for the sufficient condition, their values can be taken from the strongest direct and weakest reflective links.}:
\begin{equation}
{N_{{\text{th}}}} = \left\{ \begin{gathered}
  O\left( {{K^2}} \right), {\text{if }}\eta  < {\text{min}}\left( {\sqrt \beta  ,\alpha K} \right), 
  \text{ (user-limited)}
  \hfill \\
  O\left( {K\frac{\eta }{\alpha }} \right),{\text{if }}\alpha K < \eta  < \frac{\beta }{{\alpha K}}, 
  \text{ (gain-limited)} 
  \hfill \\
  O\left( {\frac{{{K^2}}}{\beta }{\eta ^2}} \right),{\text{otherwise}}, 
  \text{ (power-limited)}  
  \hfill \\ 
\end{gathered}  \right.
\end{equation}
where $\beta = \frac{P_r}{KP\tau^2 + \sigma_r^2}$ can be viewed as the overall RIS amplification gain.

\quad \textit{Intuitive interpretation}: When both the per-RE gain $\alpha^2$ and the overall gain $\beta$ (or reflection power $P_r$) are sufficiently large, the required number of REs $N$ depends only on the number of user pairs. When $\beta$ is sufficient but $\alpha$ is limited, each RE cannot effectively compensate for the strength disparity between the direct and cascaded reflective links, so $N$ is mainly determined by the path strength ratio $\eta$ and $\alpha$. Otherwise, $\beta$ becomes the dominant factor. Specifically, a larger $\beta$ enables more RE to operate at their maximum gain, thereby mitigating the link disparity and reducing $N$.

\item
\textbf{With total power constraints:} When the total power budget, i.e. ${P_{{\text{tot}}}} = {P_r} + N{P_{{\text{cir}}}}$ is constrained,  to guarantee interference-free transmission, ${P_{{\text{tot}}}}$ must satisfy 
\begin{equation}
{P_{{\text{tot}}}} = \Omega \left( {{K^{\frac{3}{2}}}\max \left( {{K^{\frac{1}{2}}},{\eta ^{\frac{1}{2}}}} \right)} \right).
\end{equation}
This is because if the total power ${P_{{\text{tot}}}}$ is below a certain threshold,  the required power for the REs cannot be supported, making full-DoF infeasible.
Moreover, in most cases, the required number of REs should satisfy
\begin{equation}
N = \Omega\left( {K\max \left( {K,\frac{\eta }{\alpha }} \right)} \right).
\end{equation}

\end{enumerate}

\subsection{Paper Organization and Notation}
The rest of the paper is organized as follows. Section~\ref{System_Model} introduces the system model. Section~\ref{S_IV} analyzes the required number of REs for interference-free transmission. Section~\ref{S_V}
establishes the feasibility condition for the active RIS under a limited total power budget.
Section~\ref{Simulation} presents the numerical results. Finally, Section~\ref{Conclusion} concludes the paper.

\emph{Notations:} Scalars, vectors, and matrices are denoted by lowercase ($x$), boldface lowercase (${\bf{x}}$), and boldface uppercase letters (${\bf{X}}$), respectively. The operators ${\left(  \cdot  \right)^*}$, ${\left(  \cdot  \right)^T}$, and ${\left(  \cdot  \right)^H}$ denote the conjugate, transpose, and conjugate transpose, respectively. ${\left(  \cdot  \right)^ + }$ and $\left|  \cdot  \right|$ represent the pseudo-inverse and Frobenius norm. $E\left(  \cdot  \right)$ denotes expectation. ${\text{diag}}\left( {\mathbf{x}} \right)$ returns a diagonal matrix with ${\mathbf{x}}$ on the diagonal, and ${{\mathbf{I}}_K}$ denotes the $K \times K$ identity matrix. ${\text{pro}}{{\text{j}}_{\mathbf{A}}}\left( {\mathbf{x}} \right)$ denotes the projection of ${\mathbf{x}}$ onto the column space of ${\mathbf{A}}$. $\Omega \left(  \cdot  \right)$ describes the minimal growth rate.

\section{System Model}
\label{System_Model}
As illustrated in Fig. \ref{fig1}, an active RIS with $N$ reflective elements is deployed between $K$ single-antenna user pairs. All users operate in full-duplex mode, enabling simultaneous transmission and reception. Consequently, each user $k$ receives the desired signal from its pair $k+K$ and interference from other users on both sides within the same time slot.

\begin{figure}
	\centering\includegraphics[width=6cm]{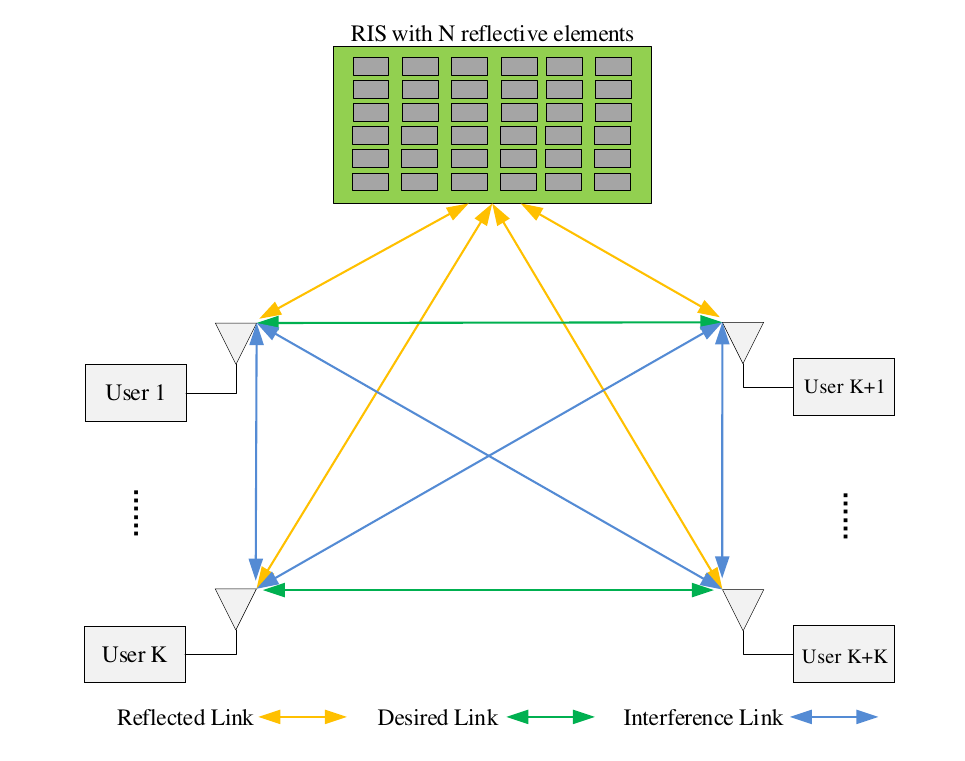}
	\caption{The RIS-aided two-way $K$-user interference channel.}
	\label{fig1}
\end{figure}

\subsection{Channel Model}
Let ${{\mathbf{h}}_{Ii}} \in {\mathbb{C}^{N \times 1}}$ and ${{\mathbf{h}}_{iI}} \in {\mathbb{C}^{N \times 1}}$ denote the channels from user $i$ to the RIS and from the RIS to user $i$, respectively. The direct channel from user $j$ to user $i$ is denoted by ${h_{ij}} \in {\mathbb{C}^{1 \times 1}}$. Channel reciprocity is assumed for all links \cite{atapattu2020reconfigurable,ma2021joint}, and perfect CSI is available at both the active RIS controller and the users \cite{peng2021multiuser,chen2023next}, which can be obtained using existing estimation methods, e.g., \cite{chen2023channel,fernandes2023channel}.

Following \cite{jiang2022interference}, which shows that the phase transition in the required number of REs ($N$) for interference nulling is independent of the channel model, we model all channels as Rayleigh fading for simplicity.
Although path loss varies with user locations, the dominant factor affecting $N$ is the intensity disparity between the direct and cascaded reflective links. To capture its impact, we assume uniform path loss\footnote{Assuming uniform path loss across users is both practical and convenient for analysis. Since stronger direct links result in more severe interference, more REs are required to suppress it. Therefore, the strongest direct link and the weakest reflective link can serve as a worst-case scenario for determining the maximum number of REs needed to guarantee interference-free transmission. Conversely, the weakest direct link and the strongest reflective link can be used to determine the minimum required number of REs. } for each link type. 
Specifically, ${{\bf{h}}_{Ii}} \sim {\cal C}{\cal N}\left( {{\bf{0}},\tau _1^2{{\bf{I}}_N}} \right)$ for $i \in \left[ {1,K} \right]$, and ${{\bf{h}}_{I{\rm{j}}}} \sim {\cal C}{\cal N}\left( {{\bf{0}},\tau _2^2{{\bf{I}}_N}} \right)$ for $j \in \left[ {K + 1,2K} \right]$, with path loss factors $\tau _1^2$ and $\tau _2^2$, respectively.  
For direct links between same-side users (i.e., $i,j \in \left[ {1,K} \right]$ or $i,j \in \left[ {K + 1,2K} \right]$), ${h_{ij}} \sim {\cal C}{\cal N}\left( {0,\sigma _2^2} \right)$, while for cross-side users
($i \in \left[ {1,K} \right]$, $j \in \left[ {K + 1,2K} \right]$), ${h_{ij}} \sim {\cal C}{\cal N}\left( {0,\sigma _3^2} \right)$. 
Scenarios with different channel models and user-specific path losses are simulated in Section \ref{Simulation}, and the results validate these assumptions.

\subsection{Signal Model}
The transmit signal of user $i$ is ${x_i} = \sqrt{P_i} s_i$, where $s_i$ is the data symbol and $P_i$ the transmit power. Each $s_i$ is an independent complex Gaussian symbol with zero mean and unit variance, so that
$E\left( {{x_i}x_j^*} \right) = \left\{ {\begin{array}{*{20}{c}}
{{P_i},i = j}\\
{0,i \ne j}
\end{array}} \right.$.

With amplifier circuits, the active RIS can control both the phase and amplitude of the incident signal. However, the thermal noise introduced by its active REs is also amplified, and thus cannot be ignored as in the passive RIS \cite{ZhouGui,9998527}. Therefore, the signal reflected from the active RIS is given by

\begin{equation}
{{\bf{x}}_{r}}{\rm{ = diag}}\left( {\bf{v}} \right)\left( {\sum\limits_{i = 1}^{2K} {{{\bf{h}}_{Ii}}{x_i}}  + {{\bf{n}}_r}} \right), 
\end{equation}
where ${{\bf{n}}_r} \sim {\cal C}{\cal N}\left( {{\bf{0}},\sigma _r^2{{\bf{I}}_N}} \right)$ denotes the thermal noise introduced by active REs. And ${\mathbf{v}} = {\left[ {{\rho _1}{e^{j{\theta _1}}}, \cdots ,{\rho _N}{e^{j{\theta _N}}}} \right]^T} \in {\mathbb{C}^{N \times 1}}$,
with ${\theta_n} \in \left[ {0,2\pi } \right)$ and ${{\rho _n}}$ representing the phase shift and amplitude coefficient applied by the $n$-th active RE, respectively. 
Although the active RIS can amplify the incident signal, its amplification is hardware-limited. Specifically, the amplitude of each entry of ${\mathbf{v}}$ must satisfy 
\begin{equation}
\left| {{\rho _n}} \right|^2 \leqslant {\alpha^2} \Leftrightarrow \left| {{v_n}} \right|^2 \leqslant {\alpha^2},\forall n \in \left[ {1,N} \right],
\end{equation}
where $\alpha^2$ denotes the maximum per-RE amplification gain \cite{10706811,long2021active}, which can exceed unity\footnote{It has been shown in \cite{long2021active,amato2018tunneling} that an active RE can achieve a gain of 40 dB (i.e., $\alpha^2 = 40{\text{dB}}$) with only 45 µW of DC power consumption. }. In addition, the total power of the amplified reflected signal is constrained by ${\left| {{{\mathbf{x}}_{r}}} \right|^2} \leqslant {P_r}$,
where ${P_r}$ is the maximum reflection power available to the active RIS \cite{9998527,9734027}. 

Due to simultaneous transmission and reception, each user experiences self-interference from its own transmit signal. 
Since each user knows its own transmitted signal and the RIS configuration, this interference can be perfectly canceled\footnote{The case where self-interference cancellation is impossible is beyond the scope of this paper and left for future work.} \cite{atapattu2020reconfigurable}. Accordingly, the received signal at user $i$ is given by

\begin{equation}
\begin{aligned}
&{y_{i}} = {\mathbf{h}}_{iI}^H{\text{diag}}\left( {\mathbf{v}} \right)\left( {\sum\limits_{j \ne i}^{2K} {{{\mathbf{h}}_{Ij}}{x_j}}  + {{\mathbf{n}}_r}} \right) + \sum\limits_{j \ne i}^{2K} {{h_{ij}}{x_j}}  + {n_i}\\
&= \left( {{\mathbf{a}}_{i\left( {K + i} \right)}^H{\mathbf{v}} + {h_{i\left( {K + i} \right)}}} \right){x_{\left( {K + i} \right)}} + {\mathbf{h}}_{iI}^H{\text{diag}}\left( {\mathbf{v}} \right){{\mathbf{n}}_r} + {n_i}\\
& + \underbrace {\sum\limits_{j \in [1,K],\hfill\atop
j \ne i\hfill} {\left( {{\bf{a}}_{ij}^H{\bf{v}} + {h_{ij}}} \right){x_j}} }_{{\rm{Same - side\, interference}}} + \underbrace {\sum\limits_{j \in [K + 1,2K],\hfill\atop
j \ne K + i\hfill} {\left( {{\bf{a}}_{ij}^H{\bf{v}} + {h_{ij}}} \right){x_j}} }_{{\rm{Cross- side \,interference}}}
\end{aligned}
\label{eq1},
\end{equation}
where ${n_i} \sim \mathcal{C}\mathcal{N}\left( {0,\sigma^2} \right)$ is the additive white Gaussian noise (AWGN) at the user $i$, ${\mathbf{a}}_{ij}^H = {\mathbf{h}}_{iI}^H\rm{diag}\left( {{{\mathbf{h}}_{Ij}}} \right) \in {\mathbb{C}^{1 \times N}}$ is the cascaded reflective channel from user $j$ to user $i$ and can be approximated as Gaussian distribution, i.e., ${{\bf{a}}_{ij}} \sim {\cal C}{\cal N}\left( {{\bf{0}},\sigma _1^2{{\bf{I}}_N}} \right)$ with ${\sigma _1} = {\tau _1}{\tau _2}$ \cite{9854102}. 

\section{Required Number of REs for Full-DoF}
\label{S_IV}
Active RIS can adjust both the phase and amplitude of incident signals, but its performance is limited by the available reflection power and the per-RE amplification gain. This section studies how these constraints affect the number of active REs required for interference-free transmission. We first formulate an equivalent feasibility problem for achieving full DoF. Then, by examining the geometric structure of the interference-nulling equations and the associated constraints, and using tools from high-dimensional probability, we derive the necessary and sufficient conditions on the number of REs.

\subsection{Equivalent Feasibility Problem}
Similar to the passive RIS, the active RIS eliminates interference by adjusting the reflected signals so that all undesired signals are destructively combined at each user. Unlike its passive counterpart, the active RIS can also amplify the incident signals. The amplification, however, is limited by each RE’s maximum gain and the total available reflection power. To account for these hardware constraints, the interference-free transmission can be expressed as
\begin{equation}
\begin{aligned}
  {\text{find }}{\mathbf{v}} \hfill \\
  {\text{s}}{\text{.t}}{\text{. }}{S_1} &= \left\{ {{\mathbf{v}}\left| {{{\mathbf{A}}^H}{\mathbf{v}} + {\mathbf{b}} = 0} \right.} \right\}, \hfill \\
  {S_3} &= \left\{ {{\mathbf{v}}\left| {{{\left| {{v_i}} \right|}^2} \leqslant \alpha^2,\forall i \in \left[ {1,N} \right]} \right.} \right\}, \hfill \\
  {S_4} &= \left\{ {{\mathbf{v}}\left| {\sum\limits_{i = 1}^{2K} {{P_i}{{\left| {{\text{diag}}\left( {\mathbf{v}} \right){{\mathbf{h}}_{Ii}}} \right|}^2}}  + \sigma _r^2{{\left| {\mathbf{v}} \right|}^2} \leqslant {P_r}} \right.} \right\},
\end{aligned}
\label{eq_feasibility}
\end{equation}
where ${S_1}$ corresponds to inter-user interference cancellation, with each row representing an interference-nulling equation (e.g., ${{\bf{a}}_{ij}^H{\bf{v}} + {h_{ij}}}=0$ ). 
Moreover, given the reciprocity of these channels (${h_{ij}} = {h_{ji}}$, ${{\bf{a}}_{ij}} = {{\bf{a}}_{ji}}$), each side (left or right) has $\frac{{K\left( {K - 1} \right)}}{2}$ distinct same-side interference nulling equations among its users, while between the two sides, there are $K\left( {K - 1} \right)$ distinct cross-interference nulling equations. Thus, each column of ${\mathbf{A}} \in {\mathbb{C}^{N \times 2K\left( {K - 1} \right)}}$ is independently distributed as ${{\mathbf{a}}_{ij}} \sim \mathcal{C}\mathcal{N}\left( {0,\sigma _1^2{{\mathbf{I}}_N}} \right)$, and ${\bf{b}} = {[{\bf{b}}_1^T,{\bf{b}}_2^T]^T} \in {\mathbb{C}^{2K\left( {K - 1} \right) \times 1}}$ with ${{\bf{b}}_1} \sim {\cal C}{\cal N}\left( {0,\sigma _2^2{{\bf{I}}_{K\left( {K - 1} \right)}}} \right)$, ${{\bf{b}}_2} \sim {\cal C}{\cal N}\left( {0,\sigma _3^2{{\bf{I}}_{K\left( {K - 1} \right)}}} \right)$. 
${S_3}$ ensures that the amplification gain of each active RE does not exceed the maximum value ${\alpha^2}$ allowed by the amplifier circuit. ${S_4}$ constrains the total power of the amplified reflective signals, including both the incident signals and the thermal noise ${{\bf{n}}_r} \sim {\cal C}{\cal N}\left( {{\bf{0}},\sigma _r^2{{\bf{I}}_N}} \right)$, to be no greater than the maximum reflection power ${{P_r}}$ available to the active RIS \cite{ZhouGui,9998527,10706811}.
Furthermore, ${S_4}$ can be rewritten as
\begin{equation}
{S_4} = \left\{ {{\mathbf{v}}\left| {{{\mathbf{v}}^H}{\mathbf{Wv}} \leqslant {P_r}} \right.} \right\},
\label{eq32}
\end{equation}
with 
\begin{equation}
\begin{aligned}
{\mathbf{W}} = {\text{diag}}\left( {\sum\limits_{i = 1}^{2K} {{{\left| {{h_{Ii,1}}} \right|}^2}{P_i}}  + \sigma _r^2, \cdots ,\sum\limits_{i = 1}^{2K} {{{\left| {{h_{Ii,N}}} \right|}^2}{P_i}}  + \sigma _r^2} \right)
\label{eq_W}, 
\end{aligned}
\end{equation}
and ${{h_{Ii,n}}}$ is the $n$-th element of ${{{\mathbf{h}}_{Ii}}}$.

Therefore, we can employ the alternating projection algorithm \cite{jiang2022interference} to solve it as follows
\begin{equation}
\left\{ \begin{gathered}
  {\mathbf{\hat v}} = {{\mathbf{v}}^{\left( t \right)}} - {\mathbf{A}}{\left( {{{\mathbf{A}}^H}{\mathbf{A}}} \right)^{ - 1}}\left( {{{\mathbf{A}}^H}{{\mathbf{v}}^{\left( t \right)}} + {\mathbf{b}}} \right) \hfill \\
  {\mathbf{\tilde v}}:{{\tilde v}_i}\mathop  = \limits^{(a)} \left\{ \begin{gathered}
  {{\hat v}_i},{\text{ if}}\left| {{{\hat v}_i}} \right| \leqslant \alpha  \hfill \\
  \alpha \frac{{{{\hat v}_i}}}{{\left| {{{\hat v}_i}} \right|}},{\text{ else}} \hfill \\ 
\end{gathered}  \right.,\forall i \in \left[ {1,N} \right] \hfill \\
  {{\mathbf{v}}^{\left( {t + 1} \right)}}\mathop  = \limits^{(b)} \left\{ \begin{gathered}
  {\mathbf{\tilde v}},{\text{ if }}{{{\mathbf{\tilde v}}}^H}{\mathbf{W\tilde v}} \leqslant {P_r} \hfill \\
  {\mathbf{\tilde v}}\sqrt {\frac{{{P_r}}}{{{{{\mathbf{\tilde v}}}^H}{\mathbf{W\tilde v}}}}} ,{\text{ else.}} \hfill \\ 
\end{gathered}  \right. \hfill \\ 
\end{gathered}  \right. 
\label{eq_AP}
\end{equation}
where ${{\mathbf{v}}^{\left( t \right)}}$ represents the vector ${\mathbf{v}}$ at the $t$-th iteration, ${\hat v}_i$ is the $i$-th element of ${\mathbf{\hat v}}$.
Operations (a) and (b) constrain the element-wise amplitude of ${\mathbf{v}}$ within $\alpha$ and the reflection power below $P_r$, satisfying ${S_3}$ and ${S_4}$, respectively.
The convergence of \eqref{eq_AP} follows a similar proof as in \cite{jiang2022interference} and is omitted for brevity.

After obtaining the optimal $\mathbf{v}$, we adopt the DoF as the performance metric, as it quantifies the number of interference-free data streams and serves as a first-order approximation of the system capacity \cite{6197715}. The system DoF is given by
\begin{equation}
    {\rm{DoF}} = \sum\limits_{i = 1}^{2K} {\mathop {\lim }\limits_{{\rm{SN}}{{\rm{R}}_i} \to \infty } \frac{{{R_i}}}{{\log \left( {{\rm{SN}}{{\rm{R}}_i}} \right)}}},
\end{equation}
where ${{\rm{SN}}{{\rm{R}}_i}}$ is the received SNR of user $i$, and the rate ${R_i}$ of user $i$ ($i \in \left[ {1,K} \right]$) is given by 
\begin{equation}
\begin{aligned}
&{R_i} = \\
&\log \left( {1 + \frac{{{P_i}{{\left| {{\mathbf{a}}_{i\left( {K + i} \right)}^H{\mathbf{v}} + {h_{i\left( {K + i} \right)}}} \right|}^2}}}{{\sum\limits_{{j \in [1,2K],} \atop {j \ne i,K + i}} {{P_j}{{\left| {{\mathbf{a}}_{ij}^H{\mathbf{v}} + {h_{ij}}} \right|}^2}}  + {{\left| {{\mathbf{h}}_{iI}^H{\text{diag}}\left( {\mathbf{v}} \right)} \right|}^2}\sigma _r^2 + {\sigma ^2}}}} \right)
\end{aligned}.
\end{equation}
Consequently, with complete interference suppression, the two-way $K$-user channel can attain a full DoF of $2K$.

\subsection{Necessary Condition on REs for Full DoF}
Note that a necessary condition for the intersection ${S_1} \cap {S_3} \cap {S_4}$ to be nonempty is that both ${S_1} \cap {S_3}$ and ${S_1} \cap {S_4}$ are nonempty. Based on this, we first decompose the feasibility problem \eqref{eq_feasibility} into two subproblems to analyze the intersections of ${S_1} \cap {S_3}$ and ${S_1} \cap {S_4}$, respectively. 
The former is addressed through the geometric properties of ${S_1}$ and ${S_3}$, while the latter is analyzed algebraically based on Gordon’s theorem \cite{gordon1988milman}. Finally, the results of these subproblems are combined to yield a necessary condition for the original problem.

Specifically, since
\begin{equation}
{S_1} \cap {S_3} \cap {S_4} \ne \varnothing  \Rightarrow {S_1} \cap {S_3} \ne \varnothing {\text{ and }}{S_1} \cap {S_4} \ne \varnothing,
\end{equation}
it follows that
\begin{equation}
{S_1} \cap {S_3} = \varnothing {\text{ or }}{S_1} \cap {S_4} = \varnothing  \Rightarrow {S_1} \cap {S_3} \cap {S_4} = \varnothing.
\end{equation}
Based on these properties, we decompose the problem \eqref{eq_feasibility} into two subproblems as follows
\begin{equation}
\begin{aligned}
{\text{find }}&{\mathbf{v}} \hfill \\
  {\text{s}}{\text{.t}}{\text{. }}&{\mathbf{v}} \in {S_1} \cap {S_3}, \hfill \\ 
\end{aligned}
\label{eq_subp1}
\end{equation}
and
\begin{equation}
\begin{aligned}
{\text{find }}&{\mathbf{v}} \hfill \\
  {\text{s}}{\text{.t}}{\text{. }}&{\mathbf{v}} \in {S_1} \cap {S_4}. \hfill \\ 
\end{aligned}
\label{eq_subp2}
\end{equation}
If the necessary conditions for subproblems \eqref{eq_subp1} and \eqref{eq_subp2} are $N \geqslant {N_1}$ and $N \geqslant {N_2}$, respectively, then we have 
\begin{equation}
N < \max \left( {{N_1},{N_2}} \right) \Rightarrow {S_1} \cap {S_3} \cap {S_4} = \varnothing
\label{eq44},
\end{equation}
\begin{equation}
{\text{or }} {S_1} \cap {S_3} \cap {S_4} \ne \varnothing  \Rightarrow N \geqslant \max \left( {{N_1},{N_2}} \right)
.
\end{equation}

Without loss of generality, we assume that all users transmit with the same power, i.e., ${P_i} = P,\forall i \in \left[ {1,2K} \right]$. Moreover, for the set ${S_4}$ given in \eqref{eq32}, since 
\begin{equation}
{{\mathbf{h}}_{Ii}} \sim \left\{ \begin{gathered}
  \mathcal{C}\mathcal{N}\left( {{\mathbf{0}},\tau _1^2{{\mathbf{I}}_N}} \right),{\text{for }}i \in \left[ {1,K} \right] \hfill \\
  \mathcal{C}\mathcal{N}\left( {{\mathbf{0}},\tau _2^2{{\mathbf{I}}_N}} \right),{\text{for }}i \in \left[ {K + 1,2K} \right] \hfill \\ 
\end{gathered}  \right.
,
\end{equation}
it follows that $\sum\limits_{i = 1}^K {{{\left| {{h_{Ii,n}}} \right|}^2}}  \sim {\text{Gamma}}\left( {K,\frac{1}{{\tau _1^2}}} \right)$.
For a large shape parameter $K$, the Gamma distribution can be well approximated by a normal distribution
\cite{grover2014application}, i.e., $\sum\limits_{i = 1}^K {{{\left| {{h_{Ii,n}}} \right|}^2}} \mathop  \sim \limits^{{\text{large }}K} \mathcal{C}\mathcal{N}\left( {K\tau _1^2,K\tau _1^4} \right)$. According to the empirical rule of normal distribution, about $99.7\% $ of the data fall within $3$ standard deviations of the mean \cite{ozdemir2016principles}. Therefore, for large $K$, almost all values of $\sum\limits_{i = 1}^K {{{\left| {{h_{Ii,n}}} \right|}^2}}$ can be expressed as 
\begin{equation}
\sum\limits_{i = 1}^K {{{\left| {{h_{Ii,n}}} \right|}^2}}  = K\tau _1^2 + c\sqrt {K\tau _1^4} \mathop  \approx \limits^{K \to \infty } K\tau _1^2,{\text{ }}\left| c \right| \leqslant 3.
\end{equation}
Thus, ${\mathbf{W}}$ in \eqref{eq_W} can be further simplified as\footnote{This approximation does not alter the order of the required number of REs and has been validated to perform well in the simulations.}
\begin{equation}
{\mathbf{W}} \mathop  \approx \limits^{K \to \infty } \left( {KP\left( {\tau _1^2 + \tau _2^2} \right) + \sigma _r^2} \right){{\mathbf{I}}_N}.
\end{equation}

Then the constraints of problem \eqref{eq_feasibility} can be rewritten as
\begin{equation}
\begin{aligned}
  &{S_1} = \left\{ {{\mathbf{v}}\left| {{\mathbf{v}} =  - {{\left( {{{\mathbf{A}}^H}} \right)}^ + }{\mathbf{b}} + {\text{null}}\left( {{{\mathbf{A}}^H}} \right)} \right.} \right\}, \hfill \\
  &{S_3} = \left\{ {{\mathbf{v}}\left| {\left| {{v_i}} \right| \leqslant \alpha ,\forall i \in \left[ {1,N} \right]} \right.} \right\}, \hfill \\
  &{S_4} = \left\{ {{\mathbf{v}}\left| {{{\mathbf{v}}^H}{\mathbf{v}} \leqslant \frac{{{P_r}}}{{KP\left( {\tau _1^2 + \tau _2^2} \right) + \sigma _r^2}} \triangleq \beta } \right.} \right\}. \hfill \\ 
\end{aligned}
\label{eq39}
\end{equation}
where ${\left( {{{\mathbf{A}}^H}} \right)^ + } = {\mathbf{A}}{\left( {{{\mathbf{A}}^H}{\mathbf{A}}} \right)^{ - 1}}$. 

Moreover, ${S_1}$ can be viewed as an affine subspace obtained by shifting the null space of ${\mathbf{A}}^H$ by the vector ${\left( {{{\mathbf{A}}^H}} \right)^+}{\mathbf{b}}$, and ${S_3}$ consists of the lines contained in the complex ball of radius $\alpha \sqrt{N}$. To determine a necessary condition for ${S_1} \cap {S_3} \ne \varnothing$, we establish the following theorem.
\begin{theorem}
Let 
${{\bar S}_1} = \left\{ {\left. {{\mathbf{x}} \in {\mathbb{C}^{N}}} \right|{\mathbf{x}} = {{\mathbf{x}}_0} + {\mathbf{z}},{\mathbf{z}} \in {\text{null}}\left( {\mathbf{G}} \right)} \right\}$, \\
${{\bar S}_3} = \left\{ {\left. {{\mathbf{x}} \in {\mathbb{C}^{N}}} \right|\left| {{x_i}} \right| \leqslant  \gamma ,\forall i = 1, \cdots ,N} \right\}$, 
where ${\mathbf{G}} \in {\mathbb{C}^{L \times N}}$ and ${{\mathbf{x}}_0} \in {\mathbb{C}^{N}}$ are given, and $L \leqslant N$.
Then a necessary condition for ${{\bar S}_1} \cap {{\bar S}_3} \ne \varnothing $ is $\left| {{\text{Pro}}{{\text{j}}_{{{\left( {{\text{null}}\left( {\mathbf{G}} \right)} \right)}^ \bot }}}\left( {{{\mathbf{x}}_0}} \right)} \right|  \leqslant  \gamma\sqrt N $.
\label{cor1}
\end{theorem}
\begin{proof}
	Please see Appendix \ref{appc1}.
\end{proof}
A geometric interpretation of Theorem~\ref{cor1} is that if the distance from the affine subspace (${{\bar S}_1}$) to the origin exceeds the ball’s radius, the subspace cannot intersect the ball or any subset (e.g, ${{\bar S}_3}$) contained in it.
Therefore, we have that  if
\begin{equation}
\begin{aligned}
\alpha \sqrt N  &<\left| {{\text{Pro}}{{\text{j}}_{{{\left( {{\text{null}}\left( {{{\mathbf{A}}^H}} \right)} \right)}^ \bot }}}\left( {{{\left( {{{\mathbf{A}}^H}} \right)}^ + }{\mathbf{b}}} \right)} \right| \\
&= \left| {{\mathbf{A}}{{\left( {{{\mathbf{A}}^H}{\mathbf{A}}} \right)}^{ - 1}}{{\mathbf{A}}^H} \cdot \left( {{\mathbf{A}}{\left( {{{\mathbf{A}}^H}{\mathbf{A}}} \right)^{ - 1}}{\mathbf{b}}} \right)} \right|
= \left| {{{\left( {{{\mathbf{A}}^H}} \right)}^ + }{\mathbf{b}}} \right|,    
\end{aligned}
\label{eq_necessary_1} 
\end{equation}
then subproblem \eqref{eq_subp1} is unsolvable (i.e., ${S_1} \cap {S_3} = \varnothing$). 
According to the norm concentration \cite{vershynin2018high}, when $N$ is large, we can approximate $\left|{{{\left( {{{\bf{A}}^H}} \right)}^ + }{\bf{b}}}\right|$ by its expected value, i.e.,
\begin{equation}
\begin{aligned}
{\left| {{{\left( {{{\bf{A}}^H}} \right)}^ + }{\bf{b}}} \right|} & \approx \sqrt {E\left( {tr\left( {{\bf{b}}{{\bf{b}}^H}{{\left( {{{\bf{A}}^H}{\bf{A}}} \right)}^{ - 1}}} \right)} \right)} \\
 &\mathop  = \limits^{(a)}  \sqrt {\frac{{K\left( {K - 1} \right)\left( {\eta _1^2 + \eta _2^2} \right)}}{{N - 2K\left( {K - 1} \right) - 1}}}, 
\end{aligned}  
\end{equation}
where ${\eta _1} = \frac{{{\sigma _2}}}{{{\sigma _1}}}$ and ${\eta _2} = \frac{{{\sigma _3}}}{{{\sigma _1}}}$ are the strength ratios of the same-side and cross-side direct links to the cascaded reflective link, respectively. The equality (a) follows from the fact that ${{{\left( {{{\mathbf{A}}^H}{\mathbf{A}}} \right)}^{ - 1}}}$ follows an inverse Wishart distribution.

Therefore, \eqref{eq_necessary_1} can be rewritten as
\begin{equation}
\alpha \sqrt N  < \left| {{{\left( {{{\mathbf{A}}^H}} \right)}^ + }{\mathbf{b}}} \right| \Leftrightarrow \alpha \sqrt N  < \sqrt {\frac{{K\left( {K - 1} \right)\left( {\eta _1^2 + \eta _2^2} \right)}}{{N - 2K\left( {K - 1} \right) - 1}}}
\label{eq_necessary_11}, 
\end{equation}
The solution of \eqref{eq_necessary_11} with respect to $N$ is given by
\begin{equation}
\begin{aligned}
N < &\frac{1}{2}\sqrt {{{\left( {2K\left( {K - 1} \right) + 1} \right)}^2} + \frac{4}{{{\alpha ^2}}}K\left( {K - 1} \right)\left( {\eta _1^2 + \eta _2^2} \right)} \\
 & + \frac{1}{2}\left( {2K\left( {K - 1} \right) + 1} \right) \triangleq {N_1}
\end{aligned}
\label{eq46}.   
\end{equation}

For the subproblem \eqref{eq_subp2}, the set $S_4$ is a ball of radius $\sqrt \beta$,  while $S_1$ is a hyperplane located at a distance $\left| {{{\left( {{{\mathbf{A}}^H}} \right)}^ + }{\mathbf{b}}} \right|$ from the origin. 
Geometrically, if $\sqrt \beta   \geqslant \left| {{{\left( {{{\mathbf{A}}^H}} \right)}^ + }{\mathbf{b}}} \right|$, the intersection ${S_1} \cap {S_4}$ is nonempty (i.e., ${S_1} \cap {S_4} \ne \varnothing$) with probability one; otherwise, if $\sqrt \beta   < \left| {{{\left( {{{\mathbf{A}}^H}} \right)}^ + }{\mathbf{b}}} \right|$, then ${S_1} \cap {S_4} = \varnothing$ with probability one. This indicates that $\sqrt \beta   = \left| {{{\left( {{{\mathbf{A}}^H}} \right)}^ + }{\mathbf{b}}} \right|$ serves as a \textit{phase transition point} determining whether $S_1$ and $S_4$ intersect.
In other words, even \textit{infinitesimal perturbations} in the values of $\beta$ or $\left| {{{\left( {{{\mathbf{A}}^H}} \right)}^ + }{\mathbf{b}}} \right|$ can trigger a discontinuous shift in the intersection status between $S_1$ and $S_4$—from nonempty to empty, or vice versa.
However, due to the randomness of ${\mathbf{A}}$ and ${\mathbf{b}}$, it is difficult to derive necessary and sufficient conditions on $N$ directly from $\sqrt \beta   = \left| {{{\left( {{{\mathbf{A}}^H}} \right)}^ + }{\mathbf{b}}} \right|$. 
To overcome this challenge, we leverage the norm concentration phenomenon \cite{vershynin2018high} in high-dimensional statistics, approximating $\left| {{{\left( {{{\mathbf{A}}^H}} \right)}^ + }{\mathbf{b}}} \right|$ by its expected value $E\left( {\left| {{{\left( {{{\mathbf{A}}^H}} \right)}^ + }{\mathbf{b}}} \right|} \right)$. 
This is justified by the fact that, as the dimension increases (i.e., $N \to \infty $), the norm $\left| {{{\left( {{{\mathbf{A}}^H}} \right)}^ + }{\mathbf{b}}} \right|$ becomes highly concentrated around its mean \cite{vershynin2018high}.

However, since the actual values of ${\left| {{{\left( {{{\mathbf{A}}^H}} \right)}^ + }{\mathbf{b}}} \right|}$ are tightly and symmetrically concentrated around the expected value $E\left( {\left| {{{\left( {{{\mathbf{A}}^H}} \right)}^ + }{\mathbf{b}}} \right|} \right)$, roughly half of the realizations exceed the expectation, while the rest fall below.
Therefore, using $\sqrt \beta   = E\left( {\left| {{{\left( {{{\mathbf{A}}^H}} \right)}^ + }{\mathbf{b}}} \right|} \right)$ to determine the required $N$ yields a condition under which the affine subspace $S_1$ intersects the ball $S_4$ with a probability close to $50\%$.
This $50\%$ phenomenon can also be found in \cite{amelunxen2014living}.
\begin{remark}
At first glance, this $50\%$ phenomenon stated here seems to contradict the derivation of the necessary conditions for ${S_1} \cap {S_3} \ne \varnothing$, in which $\left| {{{\left( {{{\mathbf{A}}^H}} \right)}^ + }{\mathbf{b}}} \right|$ is approximated by its expected value, as shown in \eqref{eq_necessary_11}. However, since ${S_3}$ is a set of 'lines' (${S_3}: \left\{ {{\mathbf{v}}\left| {\left| {{v_i}} \right| \leqslant \alpha } \right.} \right\}$) contained in complex ball (${\text{B:}}\left\{ {\left. {\mathbf{v}} \right|\left| {\mathbf{v}} \right| \leqslant \alpha \sqrt N {\text{ }}} \right\}$), the 'volume ratio' of them tends to zero, i.e., $\frac{{{\text{vol}}\left( {{S_3}} \right)}}{{{\text{vol}}\left( {\text{B}} \right)}}\mathop  \to \limits^{N \to \infty } 0$.
Thus, when $S_1$ intersects the ball with only $50\%$ probability, determined by $\alpha \sqrt N  = E\left( {\left| {{{\left( {{{\mathbf{A}}^H}} \right)}^ + }{\mathbf{b}}} \right|} \right)$, the probability that $S_1$ also intersects $S_3$ is nearly zero. Thus, the derivations of \eqref{eq_necessary_11} are reasonable.
\end{remark}
To overcome this $50\%$ phenomenon, stemming from the \textit{phase transition point} at $\sqrt \beta   = \left| {{{\left( {{{\mathbf{A}}^H}} \right)}^ + }{\mathbf{b}}} \right|$, we derive a necessary condition for the subproblem \eqref{eq_subp2} by analyzing the algebraic relationship between $\left| {{{\mathbf{A}}^H}{\mathbf{v}}} \right|$ and $\left| {\mathbf{b}} \right|$. Specifically, if the maximum magnitude of the linear combinations of the columns of  ${{{\mathbf{A}}^H}}$ over the feasible set $S_4$ is less than $\left| {\mathbf{b}} \right|$, i.e.,
\begin{equation}
\mathop {\max }\limits_{{\mathbf{v}} \in {S_4}} \left( {\left| {{{\mathbf{A}}^H}{\mathbf{v}}} \right|} \right) < \left| {\mathbf{b}} \right| 
\label{eq_necessary_2},
\end{equation}
then ${{{\mathbf{A}}^H}{\mathbf{v}} + {\mathbf{b}} = 0}$ cannot be satisfied for any ${\mathbf{v}} \in S_4$.
Consequently, ${S_1} \cap {S_4} = \varnothing$, and subproblem \eqref{eq_subp2} is unsolvable.

Since the maximum value of ${\left| {{{\mathbf{A}}^H}{\mathbf{v}}} \right|}$ over ${\mathbf{v}} \in {S_4}$ plays a key role in the subsequent derivation, we develop a modified version of Gordon’s theorem to characterize it as follows.
\begin{theorem}
Let ${\mathbf{G}} \in {\mathbb{C}^{L \times N}}$ be a random matrix with independent $\mathcal{C}\mathcal{N}\left( {0,{\sigma ^2}} \right)$ elements, $S = \left\{ {\left. {\bf{x}} \right|{{\bf{x}}^H}{\bf{x}} \leqslant \gamma  } \right\}$. Then
	\begin{equation}
	\left\{ {\begin{array}{*{20}{c}}
		{E\left( {\mathop {\min }\limits_{{\bf{x}} \in S} \left| {{\bf{Gx}}} \right|} \right) \ge \sigma \sqrt \gamma \left( {\sqrt L  - \sqrt N } \right)}\\
		{E\left( {\mathop {\max }\limits_{{\bf{x}} \in S} \left| {{\bf{Gx}}} \right|} \right) \le \sigma \sqrt \gamma \left( {\sqrt L  + \sqrt N } \right)}
		\end{array}} \right.
	\label{eq50}.    
	\end{equation}
	\label{the3}
\end{theorem}

\begin{proof}
	Please see Appendix~\ref{appC}.
\end{proof}
Moreover, since ${\bf{b}} = {[{\bf{b}}_1^T,{\bf{b}}_2^T]^T}$, where ${{\bf{b}}_1} \sim {\cal C}{\cal N}\left( {0,\sigma _2^2{{\bf{I}}_{K\left( {K - 1} \right)}}} \right)$ and ${{\bf{b}}_2} \sim {\cal C}{\cal N}\left( {0,\sigma _3^2{{\bf{I}}_{K\left( {K - 1} \right)}}} \right)$, the expectation value of its norm is 
\begin{equation}
E\left( {\left| {\mathbf{b}} \right|} \right) = \sqrt {K\left( {K - 1} \right)\left( {\sigma _2^2 + \sigma _3^2} \right)}.
\end{equation}
Therefore, according to the norm concentration, the inequality \eqref{eq_necessary_2} can be transformed as
\begin{equation}
{\sigma _1}\sqrt \beta  \left( {\sqrt {2K\left( {K - 1} \right)}  + \sqrt N } \right) < \sqrt {K\left( {K - 1} \right)\left( {\sigma _2^2 + \sigma _3^2} \right)}
\label{eq51}.
\end{equation}
The solution concerning $N$ is given by
\begin{equation}
N < {\left( {\sqrt {\frac{{K\left( {K - 1} \right)\left( {\eta _1^2 + \eta _2^2} \right)}}{\beta }}  - \sqrt {2K\left( {K - 1} \right)} } \right)^2} \triangleq {N_2}
\label{eq53},
\end{equation}
where $\beta  = \frac{{{P_r}}}{{KP\left( {\tau _1^2 + \tau _2^2} \right) + \sigma _r^2}}$ as in \eqref{eq39}, ${\eta _1} = \frac{{{\sigma _2}}}{{{\sigma _1}}}$ and ${\eta _2} = \frac{{{\sigma _3}}}{{{\sigma _1}}}$.

By substituting \eqref{eq46} and \eqref{eq53} into \eqref{eq44}, it follows that if
\begin{equation}
N < \max \left( {{N_1},{N_2}} \right) \triangleq {N_{{\text{nec}}}}
\label{eq54},
\end{equation}
where ${N_1}$ and ${N_2}$ are defined in \eqref{eq46} and \eqref{eq53}, respectively,
then ${S_1} \cap {S_3} \cap {S_4} = \varnothing$, and problem \eqref{eq_feasibility} has no solution. 
Furthermore, based on the derivations in Appendix \ref{appN_arn}, the necessary threshold ${N_{{\text{nec}}}}$ for the number of REs can be asymptotically characterized as
\begin{equation}
{N_{{\text{nec}}}} = \left\{ \begin{gathered}
  O\left( {{K^2}} \right),{\text{ if }}\eta  < {\text{min}}\left( {\sqrt \beta  ,\alpha K} \right) \hfill \\
  O\left( {\frac{K}{\alpha }\eta } \right),{\text{ if }}\alpha K < \eta  < \frac{\beta }{{\alpha K}}\hfill \\
  O\left( {\frac{{{K^2}}}{\beta }{\eta ^2}} \right),{\text{ otherwise}} \hfill \\ 
\end{gathered}  \right.
\label{eq_arn},
\end{equation}
where $\eta  = \sqrt {\eta _1^2 + \eta _2^2}$, and ${\eta _1} = \frac{{{\sigma _2}}}{{{\sigma _1}}}$ and ${\eta _2} = \frac{{{\sigma _3}}}{{{\sigma _1}}}$ can be taken as the strength ratios of the weakest same-side and cross-side direct links to the strongest cascaded reflective link, respectively. 
$\beta$ is defined in \eqref{eq39}.
If the number of REs is below this threshold, problem \eqref{eq_feasibility} is infeasible. Hence, ${N_{{\text{nec}}}}$ serves as a lower bound on the number of REs required for full DoF.

\subsection{Sufficient Condition on REs for Full DoF}
\label{SCA}
To facilitate the analysis of ${S_1} \cap {S_3} \cap {S_4}\ne \varnothing$, in this subsection, we first replace ${S_3} \cap {S_4}$ with a smaller but more tractable subset contained within it. Then, we derive a sufficient condition for ${S_1}$ to intersect this subset based on their geometric properties. Since this subset is contained in ${S_3} \cap {S_4}$, if ${S_1}$ intersects it, ${S_1}$ must also intersect ${S_3} \cap {S_4}$. Therefore, the derived condition also serves as a sufficient condition for ${S_1} \cap {S_3} \cap {S_4} \ne \varnothing$.

We first present the following theorem characterizing the intersection ${S_3} \cap {S_4}$, which underlies the subsequent derivation.
\begin{theorem}
Let ${S_3} = \left\{ {{\mathbf{v}} \in {\mathbb{C}^{ N}}\left| {\left| {{v_i}} \right| \leqslant \alpha ,\forall i \in \left[ {1, N} \right]} \right.} \right\}$ and ${S_4} = \left\{ {{\mathbf{v}} \in {\mathbb{C}^{ N}}\left| {{{\mathbf{v}}^H}{\mathbf{v}} \leqslant \beta } \right.} \right\}$, where $\alpha$ and $\beta$ are given constants. Then, the intersection ${S_3} \cap {S_4}$ satisfies:
\begin{equation}
{S_3} \cap {S_4} = \left\{ \begin{gathered}
  {S_3}{\text{, if }}\sqrt \beta   \geqslant \alpha \sqrt { N}  \hfill \\
  {S_4}{\text{, if }}\sqrt \beta   \leqslant \alpha  \hfill \\ 
\end{gathered}  \right.
\label{eq56}.
\end{equation}
Otherwise, if $\alpha  < \sqrt \beta   < \alpha \sqrt { N}$, there exists a subset 
\begin{equation}
{S_{{\text{sub}}}} = \left\{ {{\mathbf{v}} \in {\mathbb{C}^{ N}}\left| {\left| {{v_i}} \right| \leqslant \sqrt {\frac{\beta }{{ N}}} ,\forall i \in \left[ {1, N} \right]} \right.} \right\} \subseteq {S_3} \cap {S_4}.
\end{equation}

\label{the4}
\end{theorem}

\begin{proof}
Please see Appendix \ref{appE}.
\end{proof}

Consequently, a sufficient condition for ${S_1} \cap {S_3} \cap {S_4} \ne \varnothing$ can be derived by analyzing the following three cases.

\textbf{Case 1: $\sqrt \beta   \geqslant \alpha \sqrt {N}$}

In this case, since ${S_3} \cap {S_4} = {S_3}$, it follows that 
\begin{equation}
{S_1} \cap {S_3} \cap {S_4} = {S_1} \cap {S_3}.
\end{equation}
To determine a necessary condition for ${S_1} \cap {S_3} \ne \varnothing$, we establish the following theorem.
\begin{theorem}
Let
${{\bar S}_1} = \left\{ {\left. {{\mathbf{x}} \in {\mathbb{C}^{N}}} \right|{\mathbf{x}} = {{\mathbf{x}}_0} + {\mathbf{z}},{\mathbf{z}} \in {\text{null}}\left( {\mathbf{G}} \right)} \right\}$, \\
${{\bar S}_3} = \left\{ {\left. {{\mathbf{x}} \in {\mathbb{C}^{N}}} \right|\left| {{x_i}} \right| \leqslant  \gamma ,\forall i = 1, \cdots ,N} \right\}$, 
where ${\mathbf{G}} \in {\mathbb{C}^{L \times N}}$ and ${{\mathbf{x}}_0} \in {\mathbb{C}^{N}}$ are given, and $L \leqslant N$.
If $\left| {{\text{Pro}}{{\text{j}}_{{{\left( {{\text{null}}\left( {\mathbf{G}} \right)} \right)}^ \bot }}}\left( {{{\mathbf{x}}_0}} \right)} \right| \leqslant \frac{{\gamma \sqrt N }}{{\sqrt 2 }} $, then ${{\bar S}_1} \cap {{\bar S}_3} \ne \varnothing $.
\label{cor2}
\end{theorem}
\begin{proof}
	Please see Appendix \ref{appc2}.
\end{proof}

Thus, according to Theorem~\ref{cor2} and \eqref{eq_necessary_1}, we have
\begin{equation}
\begin{aligned}
  &\alpha \sqrt N \geqslant \sqrt 2 \left| {{{\left( {{{\mathbf{A}}^H}} \right)}^ + }{\mathbf{b}}} \right| \\
  &\Leftrightarrow \sqrt {\frac{{K\left( {K - 1} \right)\left( {\eta _1^2 + \eta _2^2} \right)}}{{N - 2K\left( {K - 1} \right) - 1}}}  \leqslant \frac{{\alpha \sqrt N }}{{\sqrt 2 }}
   \Rightarrow {S_1} \cap {S_3} \ne \varnothing. \hfill \\ 
\end{aligned}
\end{equation}
Consequently, the sufficient condition for ${S_1} \cap {S_3} \cap {S_4} \ne \varnothing$ in this case can be expressed as
\begin{equation}
\begin{aligned}
  N \geqslant \frac{1}{2}\left( {a + \sqrt {{a^2} + \frac{8}{{{\alpha ^2}}}b} } \right) \triangleq {N_{{\text{suf1}}}} \hfill \\ 
\end{aligned}
\label{eq60},
\end{equation}
where $a = 2K\left( {K - 1} \right) + 1$ and $b = K\left( {K - 1} \right)\left( {\eta _1^2 + \eta _2^2} \right)$.

\textbf{Case 2: $\alpha  < \sqrt \beta   < \alpha \sqrt {N}$}

In this case, since ${S_{{\text{sub}}}} \subseteq {S_3} \cap {S_4}$, we have
\begin{equation}
{S_1} \cap {S_{{\text{sub}}}} \ne \varnothing  \Rightarrow {S_1} \cap {S_3} \cap {S_4} \ne \varnothing.
\end{equation}
Similarly, according to Theorem~\ref{cor2}, we have
\begin{equation}
\sqrt {\frac{{K\left( {K - 1} \right)\left( {\eta _1^2 + \eta _2^2} \right)}}{{N - 2K\left( {K - 1} \right) - 1}}}  \leqslant \sqrt {\frac{\beta }{2}}  \Rightarrow {S_1} \cap {S_{{\text{sub}}}} \ne \varnothing.
\end{equation}
Therefore, the sufficient condition for ${S_1} \cap {S_3} \cap {S_4} \ne \varnothing$ in this case is given by
\begin{equation}
\begin{aligned}
N \geqslant \frac{{2b}}{\beta } + a  \triangleq {N_{{\text{suf2}}}}
\end{aligned}
\label{eq63},
\end{equation}
where $a = 2K\left( {K - 1} \right) + 1$ and $b = K\left( {K - 1} \right)\left( {\eta _1^2 + \eta _2^2} \right)$.

\textbf{Case 3: $\sqrt \beta   \leqslant \alpha$}

In this case, since ${S_3} \cap {S_4} = {S_4}$, it follows that
\begin{equation}
{S_1} \cap {S_3} \cap {S_4} = {S_1} \cap {S_4}.
\end{equation}
Moreover, because ${S_{{\text{sub}}}} \subseteq {S_4}$, the sufficient condition \eqref{eq63} derived in Case 2 also holds in this case.

Therefore, by combining the results of the three cases and based on the derivations in Appendix \ref{appF}, a sufficient condition on the required number of  REs (i.e., $N$) to ensure ${S_1} \cap {S_3} \cap {S_4} \ne \varnothing$ can be expressed as
\begin{equation}
N \geqslant {N_{{\text{suf}}}} \triangleq \left\{ \begin{gathered}
  {N_{{\text{suf1}}}}{\text{, if }}\frac{\beta }{{{\alpha ^2}}} \geqslant {N_{{\text{suf1}}}} \hfill \\
  {N_{{\text{suf2}}}},{\text{ otherwise}} \hfill \\ 
\end{gathered}  \right.
\label{eq65},
\end{equation}
where $\beta=\frac{{{P_r}}}{{KP\left( {\tau _1^2 + \tau _2^2} \right) + \sigma _r^2}}$ as shown in \eqref{eq39}, ${N_{{\text{suf1}}}}$ and ${N_{{\text{suf2}}}}$ are defined in \eqref{eq60} and \eqref{eq63}, respectively.

Furthermore, based on the derivations in Appendix \ref{appN_ars}, the sufficient thresholds ${N_{{\text{suf}}}}$ in \eqref{eq65} can be asymptotically characterized as
\begin{equation}
{N_{{\text{suf}}}}  = \left\{ \begin{gathered}
  O\left( {{K^2}} \right),{\text{ if }}\eta  < {\text{min}}\left( {\sqrt \beta  ,\alpha K} \right) \hfill \\
  O\left( {\frac{K}{\alpha }\eta } \right),{\text{ if }}\alpha K < \eta  < \frac{\beta }{{\alpha K}}{\text{ }} \hfill \\
  O\left( {\frac{{{K^2}}}{\beta }{\eta ^2}} \right),{\text{ otherwise}} \hfill \\ 
\end{gathered}  \right.
\label{eq_ars},
\end{equation}
where $\eta  = \sqrt {\eta _1^2 + \eta _2^2}$, and ${\eta _1} = \frac{{{\sigma _2}}}{{{\sigma _1}}}$ and ${\eta _2} = \frac{{{\sigma _3}}}{{{\sigma _1}}}$ can be taken as the strength ratios of the strongest same-side and cross-side direct links to the weakest cascaded reflective link, respectively. 
If the number of active REs exceeds this threshold, the feasibility of problem \eqref{eq_feasibility} can be guaranteed. Hence, ${N_{{\text{suf}}}}$ serves as a minimal upper bound on the required number of REs for achieving the full DoF.

\section{Reflection–Circuit Power Tradeoff}
\label{S_V}
As analyzed above, for a given reflection power and per-RE amplification gain, interference-free transmission can be achieved by increasing the number of REs. However, in practice, the hardware power consumption of active RE circuits increases with their number, while the total power budget remains limited. Consequently, the power available for amplified reflection decreases as more active REs are deployed, potentially making interference-free transmission infeasible. In this section, we derive the minimum total power and number of REs required to ensure interference-free transmission.

Since the total power consumption of an active RIS is given by ${P_{{\text{tot}}}} = {P_{{\text{ref}}}} + N{P_{{\text{cir}}}}$ \cite{9734027},
where ${P_{{\text{ref}}}}$ denotes the power of the amplified reflected signal, and ${P_{{\text{cir}}}} = {P_{{\text{SW}}}} + {P_{{\text{DC}}}}$ represents the fixed hardware power consumption per active RE. Here, ${P_{{\text{SW}}}}$ is the power consumed by the phase shifter switch and control circuit, and ${P_{{\text{DC}}}}$ is the DC bias power required by the amplifier of each active RE.  
Therefore, for a given total power budget ${P_{{\text{tot}}}}$, the maximum available reflection power ${P_r}$ is
\begin{equation}
{P_r} = {P_{{\text{tot}}}} - N{P_{{\text{cir}}}}.
\end{equation}

Consequently, the sufficient condition for the required number of active REs under a given total power budget ${P_{{\text{tot}}}}$ can be obtained by replacing ${P_r}$ in \eqref{eq65} with ${P_{{\text{tot}}}} - N{P_{{\text{cir}}}}$, i.e.,
\begin{equation}
N \geqslant \left\{ \begin{gathered}
  {\bar N_{{\text{suf1}}}}{\text{, if }}\frac{{{P_{{\text{tot}}}} - N{P_{{\text{cir}}}}}}{{\mu {\alpha ^2}}} \geqslant {\bar N_{{\text{suf1}}}} \hfill \\
  {\bar N_{{\text{suf2}}}},{\text{ otherwise}} \hfill \\ 
\end{gathered}  \right.
\label{eq69},
\end{equation}
where ${\bar N_{{\text{suf1}}}} = \frac{1}{2}\left( {a + \sqrt {{a^2} + \frac{8}{{{\alpha ^2}}}b} } \right)$, ${\bar N_{{\text{suf2}}}} = \frac{{2b\mu }}{{{P_{{\text{tot}}}} - N{P_{{\text{cir}}}}}} + a$, $\mu  = KP\left( {\tau _1^2 + \tau _2^2} \right) + \sigma _r^2$, $a = 2K\left( {K - 1} \right) + 1$ and $b = K\left( {K - 1} \right)\left( {\eta _1^2 + \eta _2^2} \right)$. ${\eta _1}$ and ${\eta _2}$ can be taken as the strength ratios of the strongest same-side and cross-side direct links to the weakest cascaded reflective link, respectively.
However, since the piecewise condition involves the unknown parameter $N$, \eqref{eq69} cannot be applied directly.

To address this issue, and based on the derivations in Appendix \ref{appG}, the sufficient condition \eqref{eq69} can be rewritten as
\begin{equation}
{{\bar N}_2} \geqslant N \geqslant\left\{ \begin{gathered}
   {\bar N_{{\text{suf1}}}},{\text{ if }}{P_{{\text{tot}}}} \geqslant \left( {{P_{{\text{cir}}}} + \mu {\alpha ^2}} \right){\bar N_{{\text{suf1}}}} \hfill \\
 {{\bar N}_1},{\text{ if }}\left( {{P_{{\text{cir}}}} + \mu {\alpha ^2}} \right){\bar N_{{\text{suf1}}}} > {P_{{\text{tot}}}} \geqslant {P_\vartriangle } \hfill \\ 
\end{gathered}  \right.
\label{eq70},
\end{equation}
where ${{\bar N}_{\text{1}}} = \frac{{{P_{{\text{tot}}}} + a{P_{{\text{cir}}}} - \sqrt \Delta  }}{{2{P_{{\text{cir}}}}}}$, ${{\bar N}_2} = \frac{{{P_{{\text{tot}}}} + a{P_{{\text{cir}}}} + \sqrt \Delta  }}{{2{P_{{\text{cir}}}}}}$, ${P_\vartriangle } = a{P_{{\text{cir}}}} + \sqrt {8{P_{{\text{cir}}}}b\mu }$, and $\Delta  = {\left( {{P_{{\text{tot}}}} + a{P_{{\text{cir}}}}} \right)^2} - 4{P_{{\text{cir}}}}\left( {2b\mu  + a{P_{{\text{tot}}}}} \right)$.

Therefore, the minimal total power should satisfy ${P_{{\text{tot}}}} \geqslant {P_\Delta } = O\left( {{K^{\frac{3}{2}}}\max \left( {{K^{\frac{1}{2}}},{\eta ^{\frac{1}{2}}}} \right)} \right)$, and $\frac{{\partial {{\bar N}_{\text{1}}}}}{{\partial {P_{{\text{tot}}}}}} < 0$.
${{\bar N}_2}$ and ${\bar N_{{\text{suf1}}}}$ can be asymptotically
characterized as 
${{\bar N}_2} =O\left( \frac{P_\text{tot}}{P_\text{cir}} \right)$ and
\begin{equation}
{\bar N_{{\text{suf1}}}} =O\left( {K\max \left( {K,\frac{\eta }{\alpha }} \right)} \right).
\end{equation}

The sufficient condition in \eqref{eq70} offers valuable guidance for system design to ensure reliable performance. Specifically, if the total power budget ${P_{\text{tot}}}$ is insufficient (i.e., ${P_{\text{tot}}} < {P_\vartriangle }$), interference-free transmission cannot be guaranteed, regardless of how the active RIS is configured. Moreover, even with adequate power, interference-free transmission still cannot be guaranteed if the number of REs falls outside a certain interval. This is because too many active REs increase circuit power, thereby reducing the power available for reflection. While too few REs cannot suppress interference effectively.

\section{Simulation Result}
\label{Simulation}
In this section, we present numerical results to validate our theoretical findings on the required number of active REs to achieve interference-free transmission or the full DoF $2K$ in the two-way $K$-user interference channel.

\subsection{Simulation Setup}
Without loss of generality, all channels are assumed reciprocal. Unless stated otherwise, the system bandwidth is $1\rm{MHz}$ and the noise spectral density is $-174\rm{dBm/Hz}$. The noise powers at the RIS and users are equal \cite{long2021active, ZhouGui}, i.e., ${\sigma^2} = \sigma_r^2$, and all users transmit with $P = 30\rm{dBm}$. The left and right side users are evenly distributed along the $x$-axis within $\left[ {5\rm{m},50\rm{m}} \right]$ at $y=-15\rm{m}$ and $y=15\rm{m}$, respectively, both at $z=-20\rm{m}$. The RIS is located at $(0,0,0)$.

To assess the practical applicability of our theoretical results, the RIS-related channels are modeled as Rician fading
\begin{equation}
{\bf{\tilde H}} = \sqrt L\left( {\sqrt {\frac{\varepsilon }{{\varepsilon  + 1}}} {{{\bf{\tilde H}}}^{{\rm{LoS}}}} + \sqrt {\frac{1}{{\varepsilon  + 1}}} {{{\bf{\tilde H}}}^{{\rm{NLoS}}}}} \right),
\end{equation}
where $\varepsilon=10$ is the Rician factor, and ${{\bf{\tilde H}}^{{\rm{LoS}}}}$ and ${{\bf{\tilde H}}^{{\rm{NLoS}}}}$ represent the line-of-sight (LoS) and non-LoS (NLoS) components, respectively, as in \cite{jiang2021learning}. The direct user-to-user links remain Rayleigh fading due to extensive scattering \cite{pan2020multicell}.
The large-scale path loss $L$ in dB is modeled as $L\left( \text{d} \right) =  - 30 - 10\omega {\log _{10}}\left( \text{d} \right)$, where $\text{d}$ is the link distance in meters \cite{jiang2021learning,ZhouGui}. The path-loss exponent $\omega$ is set to 2 for reflective links and 4 for direct links \cite{oa2018determination}. And 200 independent trials are conducted.

\subsection{Required REs under Different Hardware Constraints}
\label{A_RE}
This subsection evaluates the number of active REs required for full DoF, as analyzed in Section~\ref{S_IV}. In Figs.~\ref{alpha}–\ref{user_P}, the curves labeled “$99\%$” and “$1\%$ probability of full-DoF” indicate the values of $N$ achieving full DoF with $99\%$ and $1\%$ probability, respectively. The “necessary” line is derived from \eqref{eq54}, where ${\eta _1}$ and ${\eta _2}$ are computed by the average path losses of the cascaded reflective, same-side direct, and opposite-side direct links. The reflective path losses $\tau_1^2$ and $\tau_2^2$ are set to the square roots of the cascaded loss. The “sufficient” line comes from \eqref{eq65}, where ${\eta _1}$ and ${\eta _2}$ are computed based on the average path losses of the same-side and opposite-side direct links, while the cascaded reflective path loss is set to its minimum value determined by users at $\left( {50{\text{m}}, \pm 15{\text{m}}, -20{\text{m}}} \right)$.

Fig.\ref{alpha} illustrates the relationship between the maximum per-RE amplification gain ${\alpha^2}$ and the number of REs required for full DoF. It shows that the required number of REs is nearly independent of the channel model. Moreover, as ${\alpha^2}$ increases, the required number of REs decreases, since stronger amplification enhances the ability of RIS to reflect and cancel interference. However, when ${\alpha^2}$ exceeds a certain threshold, further increases do not reduce the number of REs, which remains higher than $N = 2K(K - 1)$ (i.e., the number of interference-nulling equations)  \cite{10706811}. This is because the reflection power ${P_r}=30\mathrm{dBm}$ is relatively insufficient at high amplification levels and thus becomes the main factor determining the required number of REs. In contrast, when both ${\alpha^2}$ and ${P_r}$ are sufficiently large, the required number of REs closely matches $N = 2K(K - 1)$ that depends only on the number of user pairs $K$, as demonstrated in Fig.~\ref{Pr}.
\begin{figure}
	\centering\includegraphics[width=6.5cm]{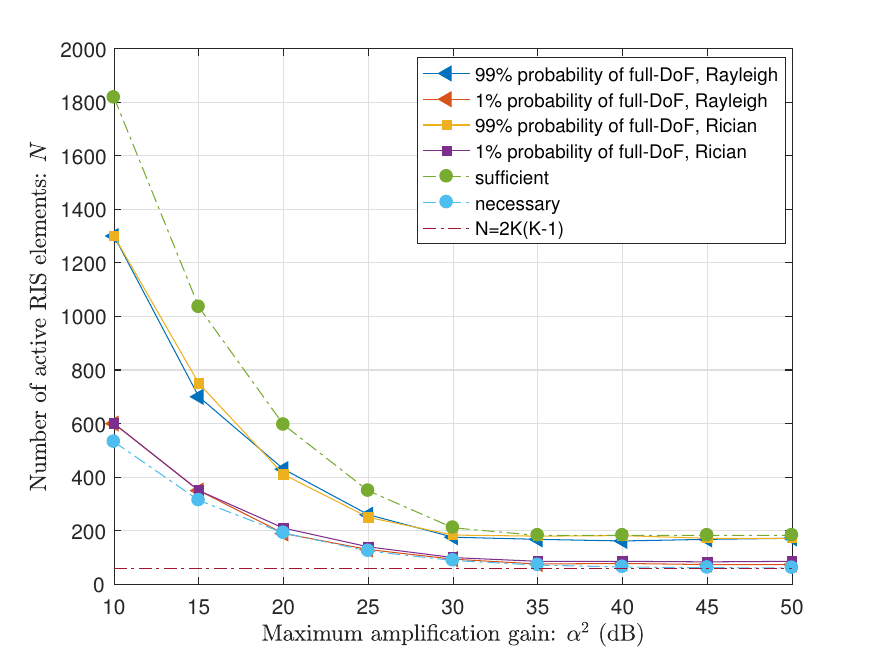}
    \caption{Required REs for full DoF vs. maximum amplification gain. ($K=6$, ${P_r} = 30{\text{dBm}}$,  ${P} = 30{\text{dBm}}$).}
	\label{alpha}
\end{figure}

Fig.~\ref{Pr} illustrates the relationship between the reflection power $P_r$ and the number of REs required for full DoF. It is observed that when the per-RE amplification gain ${\alpha^2}$ is sufficiently large but $P_r$ is small, the required number of REs decreases with increasing $P_r$. However, when $P_r$ exceeds a certain threshold such that both $P_r$ and ${\alpha^2}$ are relatively sufficient, the required number of REs is approximately $N = 2K(K - 1)$, which depends only on the number of user pairs $K$. It is consistent with our theoretical analysis.

\begin{figure}
	\centering\includegraphics[width=6.5cm]{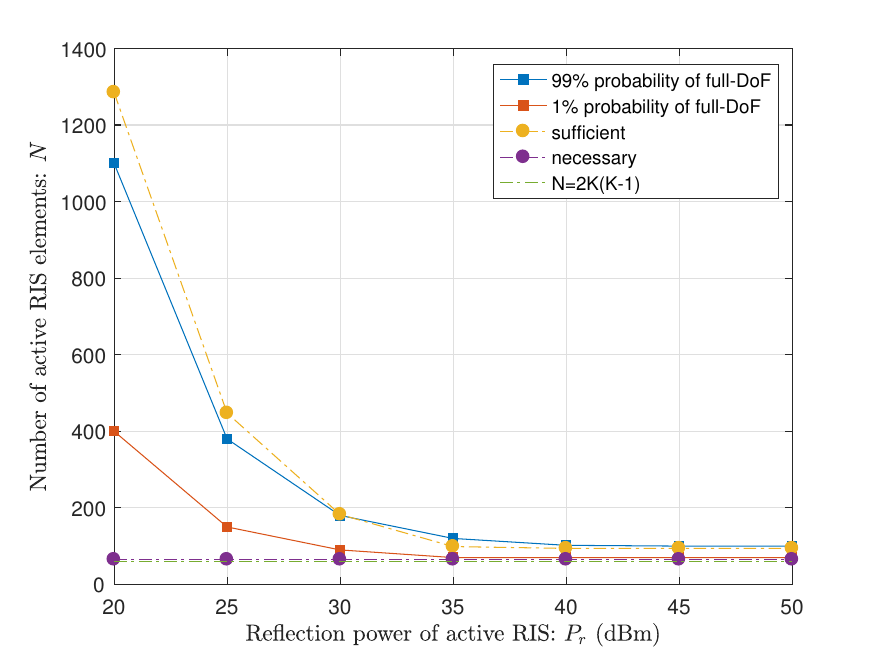}
    \caption{Required REs for full DoF vs. reflection power ($K=6$, $P = 30{\text{dBm}}$, ${\alpha ^2} = 40{\text{dB}}$).}
	\label{Pr}
\end{figure}

Fig.~\ref{user_P} illustrates the relationship between the user transmission power $P$ and the number of REs required for full DoF.  
It can be observed that when $P$ is small, the required number of REs remains constant as $P$ increases. However, this number is larger than $N = 2K(K - 1)$. 
This is because, in this case, although $P_r$ is sufficient, ${\alpha^2}$ is set to 20dB, which is insufficient and becomes the dominant factor determining the required number of REs. In contrast, when $P$ exceeds a certain threshold, the number of active REs increases with $P$, as both ${\alpha^2}$ and $P_r$ become insufficient relative to the higher user transmit power. Consequently, more REs are needed to mitigate the intensified interference. These trends are consistent with our theoretical analysis.

\begin{figure}
	\centering\includegraphics[width=6.5cm]{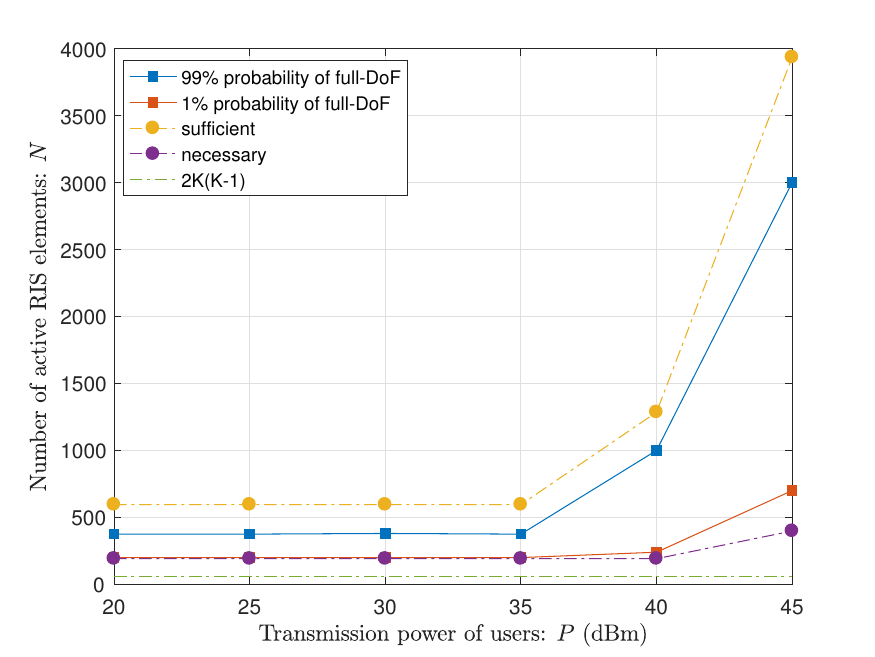}
    \caption{Required REs for full DoF vs. user transmission power ($K=6$, ${P_r} = 30{\text{dBm}}$, ${\alpha^2} = 20{\text{dB}}$).}
	\label{user_P}
\end{figure}

\subsection{Feasible RE Range under a Limited Power Budget}
This subsection evaluates the number of REs required for full DoF under a fixed total power. In Fig.~\ref{tradeoff}, the “Theoretical range” is given by \eqref{eq70}, with the channel strength ratios (${\eta_1}$, ${\eta_2}$) and reflective path losses ($\tau_1^2$, $\tau_2^2$) set as in the “sufficient” line of Subsection~\ref{A_RE}. For comparison, the “Conventional range” is given by $N_\text{low}=2K(K-1)$ and $N_\text{up}=P_\text{tot}/P_\text{cir}$, where $N_\text{low}$ corresponds to the number of interference-nulling equations \cite{10706811}, and $N_\text{up}$ represents the extreme case in which all power is allocated to the amplification circuits.

It shows that neither too few nor too many active REs can achieve full DoF. This is because, under a fixed total power budget, too few REs cannot effectively suppress interference, while too many REs consume excessive circuit power, leaving inadequate power for reflection. Moreover, the RE's range that ensures full DoF closely aligns with the theoretical prediction. 

\begin{figure}
	\centering\includegraphics[width=6.5cm]{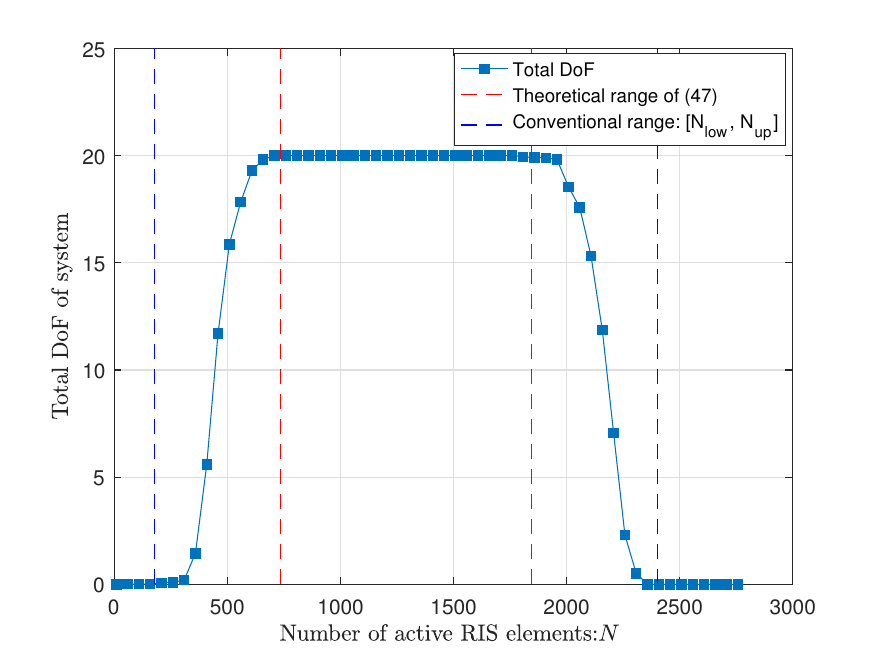}
    \caption{Number of REs vs. total DoF ($K=10$, ${P_{{\text{tot}}}}=30{\text{dBm}}$, ${P_{{\text{SW}}}}=-10{\text{dBm}}$, ${P_{{\text{DC}}}}=-5{\text{dBm}}$, ${\alpha ^2} = 40{\text{dB}}$, $P=20{\text{dBm}}$).}
	\label{tradeoff}
\end{figure}

\section{Conclusion}
\label{Conclusion}
This paper studies interference-free transmission with active RIS in the two-way $K$-user interference channel. Due to hardware limits on per-RE gain and total reflection power, achieving interference-free transmission is equivalent to checking the feasibility of a random linear system with joint norm and box constraints. By decomposing this problem into set-intersection subproblems and approximating the intersection of the joint constraints with a smaller, tractable set, necessary and sufficient conditions for interference-free transmission are derived based on geometric and probabilistic analysis.

The results show that when both the per-RE gain and reflection power are sufficient, the required number of active REs depends only on the user pairs. If reflection power is sufficient but gain is limited, the number of REs depends on both the gain and the channel strength ratios; otherwise, it is mainly inversely proportional to the reflection power. With a limited total power budget, the number of REs must lie within a certain range to ensure interference-free transmission. The proposed framework provides a general method for analyzing other feasibility problems in constrained random linear systems.

\appendix
\subsection{Proof of Theorem \ref{cor1}}	
\label{appc1}
Since each $\mathbf{x} \in \bar{S}_1$ can be decomposed as
\begin{equation}
{\mathbf{x}} = {{\mathbf{\tilde x}}} + {{\mathbf{x}}_\parallel },    
\end{equation}
where ${{\mathbf{\tilde x}}} = {\text{Pro}}{{\text{j}}_{{{\left( {{\text{null}}\left( {\mathbf{G}} \right)} \right)}^ \bot }}}\left( {{{\mathbf{x}}_0}} \right)$ represents the projection of ${{\mathbf{x}}_0}$ onto the orthogonal complement of ${{\text{null}}\left( \mathbf{G} \right)}$.
${{\mathbf{x}}_\parallel }$ is a vector lying in the null space of ${\mathbf{G}}$ (i.e., ${{\mathbf{x}}_\parallel } \in {\text{null}}\left( {\mathbf{G}} \right)$), and ${{\mathbf{\tilde x}}} \bot {{\mathbf{x}}_\parallel }$. 
Thus, if ${{\bar S}_1} \cap {{\bar S}_3} \ne \varnothing$, then there must exist a ${{\mathbf{x}}_\parallel }$ such that
\begin{equation}
{\left| {{{\mathbf{\tilde x}}} + {{\mathbf{x}}_\parallel }} \right|^2} \leqslant {\gamma ^2}N \Leftrightarrow {\left| {{\mathbf{\tilde x}}} \right|^2} + {\left| {{{\mathbf{x}}_\parallel }} \right|^2} \leqslant {\gamma ^2}N.  
\label{eqa48}
\end{equation}
Thus, a feasible ${{\mathbf{x}}_\parallel }$ may exist (i.e., ${{\bar S}_1} \cap {{\bar S}_3} \ne \varnothing$) only when
\begin{equation}
{\left| {{\mathbf{\tilde x}}} \right|^2} \leqslant {\gamma ^2}N \Leftrightarrow \left| {{\text{Pro}}{{\text{j}}_{{{\left( {{\text{null}}\left( {\mathbf{G}} \right)} \right)}^ \bot }}}\left( {{{\mathbf{x}}_0}} \right)} \right| \leqslant \gamma\sqrt N .
\end{equation}

\subsection{Proof of Theorem \ref{the3}}	
\label{appC}
Since each element of ${\mathbf{G}} \in {\mathbb{C}^{L \times N}}$ is i.i.d.$ \sim \mathcal{C}\mathcal{N}\left( {0,{\sigma ^2}} \right)$, and ${\mathbf{x}} \in S = \left\{ {\left. {\bf{x}} \right|{{\bf{x}}^H}{\bf{x}} \leqslant \gamma  } \right\}$, the random vector ${{\mathbf{Gx}}}$ has the same distribution as $\sigma \sqrt \gamma {\mathbf{\tilde G\tilde x}}$, i.e., ${\mathbf{Gx}} \mathop  = \limits^{\text{d}} \sigma \sqrt \gamma {\mathbf{\tilde G\tilde x}}$,
where each element of ${\mathbf{\tilde G}} \in {\mathbb{C}^{L \times N}}$ is i.i.d.$ \sim \mathcal{C}\mathcal{N}\left( {0,1} \right)$ , and 
\begin{equation}
\begin{array}{l}
{\mathbf{\tilde x}} \in \tilde S = \left\{ {\left. {{\mathbf{\tilde x}}} \right|{{{\mathbf{\tilde x}}}^H}{\mathbf{\tilde x}}  \leqslant 1} \right\}.
\end{array}
\end{equation}
Thus we have 
\begin{equation}
E\left( {\left| {{\mathbf{Gx}}} \right|} \right) = \sigma \sqrt \gamma  \cdot E\left( {\left| {{\mathbf{\tilde G\tilde x}}} \right|} \right)
\label{eqa63}.
\end{equation}
According to Gordon's Theorem \cite{gordon1988milman}, we can get
\begin{equation}
\left\{ \begin{gathered}
  E\left( {\mathop {\min }\limits_{{\mathbf{\tilde x}} \in \tilde S} \left| {{\mathbf{\tilde G\tilde x}}} \right|} \right) \geqslant {a_L} - W\left( {\tilde S} \right) \hfill \\
  E\left( {\mathop {\max }\limits_{{\mathbf{\tilde x}} \in \tilde S} \left| {{\mathbf{\tilde G\tilde x}}} \right|} \right) \leqslant {a_L} + W\left( {\tilde S} \right) \hfill \\ 
\end{gathered}  \right.
\label{eqa64},
\end{equation}
where ${a_L} = \sqrt L $, and $W\left( {\tilde S} \right)$ denotes the Gaussian width \cite{vershynin2018high} of the set ${\tilde S}$ which can be calculated as
\begin{equation}
\begin{aligned}
W\left( {\tilde S} \right) &= E\left( {\mathop {\sup }\limits_{{\mathbf{\tilde x}} \in \tilde S} \left| {\left\langle {{\mathbf{g}},{\mathbf{\tilde x}}} \right\rangle } \right|} \right) \hfill
= E\left( {\left\{ {\begin{aligned}
		&\max \left| {{{\mathbf{g}}^H}{\mathbf{\tilde x}}} \right|\\
		&s.t.{\text{ }}{{{\mathbf{\tilde x}}}^H}{\mathbf{\tilde x}} \leqslant 1
		\end{aligned}} \right.} \right) \\
& \mathop  = \limits^{\left( a \right)} E\left( {\mathop {\max }\limits_{{\mathbf{\tilde x}} \in \tilde S} \left( {\left| {\mathbf{g}} \right| \cdot \left| {{\mathbf{\tilde x}}} \right|} \right)} \right)       
=  E\left( {\left| {\mathbf{g}} \right|} \right)  \hfill 
\mathop  = \limits^{N \to \infty } \sqrt N 
\end{aligned}
\label{eqa65}
\end{equation}
where  ${\mathbf{g}} \sim \mathcal{C}\mathcal{N}\left( {0,{{\mathbf{I}}_N}} \right)$ is a gaussian vector, and (a) follows from the Cauchy–Schwarz inequality.
Combining \eqref{eqa63} and \eqref{eqa64} we can get \eqref{eq50},
which completes the proof.

\subsection{Proof of the Asymptotic Form of ${N_{{\text{nec}}}}$ in \eqref{eq_arn}}
\label{appN_arn}
According to the definitions of $N_1$ and $N_2$ given in  \eqref{eq46} and \eqref{eq53}, respectively, we have
\begin{equation}
{N_1} = \left\{ \begin{gathered}
  O\left( {{K^2}} \right),{\text{ if }}\eta  < \alpha K \hfill \\
  O\left( {\frac{K}{\alpha }\eta } \right),{\text{ if }}\eta  > \alpha K \hfill \\ 
\end{gathered}  \right.
,
\end{equation}
\begin{equation}
{N_2} = \left\{ \begin{gathered}
  O\left( {{K^2}} \right),{\text{ if }}\eta  < \sqrt \beta   \hfill \\
  O\left( {\frac{{{K^2}}}{\beta }{\eta ^2}} \right),{\text{ if }}\eta  > \sqrt \beta   \hfill \\ 
\end{gathered}  \right.
.
\end{equation}
If $\sqrt \beta   < \alpha K$, then ${N_{{\text{nec}}}} \triangleq \max \left( {{N_1},{N_2}} \right)$ can be rewritten as
\begin{equation}
{N_{{\text{nec}}}} = \left\{ \begin{gathered}
  O\left( {{K^2}} \right),{\text{ if }}\eta  < \sqrt \beta   \hfill \\
  O\left( {\frac{{{K^2}}}{\beta }{\eta ^2}} \right),{\text{ if }}\eta  > \sqrt \beta   \hfill \\ 
\end{gathered}  \right.
\label{arn_1}.
\end{equation}
Moreover, if $\sqrt \beta   > \alpha K$, then we have
\begin{equation}
{N_{{\text{nec}}}} = \left\{ \begin{gathered}
  O\left( {{K^2}} \right),{\text{ if }}\eta  < \alpha K \hfill \\
  O\left( {\frac{K}{\alpha }\eta } \right),{\text{ if }}\alpha K < \eta  < \frac{\beta }{{\alpha K}} \hfill \\
  O\left( {\frac{{{K^2}}}{\beta }{\eta ^2}} \right),{\text{ if }}\frac{\beta }{{\alpha K}} < \eta  \hfill \\ 
\end{gathered}  \right.
\label{arn_2}.
\end{equation}
Combining \eqref{arn_1} and \eqref{arn_2}, it follows that when
\begin{equation}
\begin{aligned}
   \eta  < \sqrt \beta   < \alpha K {\text{ or }}\eta  < \alpha K < \sqrt \beta
   \Leftrightarrow \eta  < \min \left( {\alpha K,\sqrt \beta  } \right)
\end{aligned}
,
\end{equation}
we have ${N_{{\text{nec}}}} = O\left( {{K^2}} \right)$. Additionally, when
\begin{equation}
\sqrt \beta   > \alpha K{\text{ and }}\alpha K < \eta  < \frac{\beta }{{\alpha K}},
\end{equation}
which is equivalent to $\alpha K < \eta  < \frac{\beta }{{\alpha K}}$ as $\alpha K < \frac{\beta }{{\alpha K}} \Leftrightarrow \alpha K < \sqrt \beta$,
we have ${N_{{\text{nec}}}} = O\left( {\frac{K}{\alpha }\eta } \right)$. For the remaining cases, ${N_{{\text{nec}}}} = O\left( {\frac{{{K^2}}}{\beta }{\eta ^2}} \right)$. This completes the proof.

\subsection{Proof of Theorem \ref{the4}}	
\label{appE}
Since ${S_3} = \left\{ {{\mathbf{v}} \in {\mathbb{C}^{ N}}\left| {\left| {{v_i}} \right| \leqslant \alpha ,\forall i \in \left[ {1, N} \right]} \right.} \right\}$ and ${S_4} = \left\{ {{\mathbf{v}} \in {\mathbb{C}^{ N}}\left| {{{\mathbf{v}}^H}{\mathbf{v}} \leqslant \beta } \right.} \right\}$,
for any ${\mathbf{v}} \in {S_3}$, we have 
\begin{equation}
{{\mathbf{v}}^H}{\mathbf{v}} = \sum\limits_{i = 1}^{N} {{{\left| {{v_i}} \right|}^2}}  \leqslant  N{\alpha ^2}.
\end{equation}
If $\beta  \geqslant N{\alpha ^2}$, it follows that ${S_3} \subseteq {S_4}$ and hence ${S_3} \cap {S_4} = {S_3}$.

Conversely, for any ${\mathbf{v}} \in {S_4}$, since ${{{\mathbf{v}}^H}{\mathbf{v}} \leqslant \beta }$, it follows that $\left| {{v_i}} \right| \leqslant \sqrt \beta  {\text{, }}\forall i$. Thus, if $\alpha  \geqslant \sqrt \beta$, then ${S_4} \subseteq {S_3}$, and hence ${S_3} \cap {S_4} = {S_4}$.

For the remaining case $\alpha  < \sqrt \beta   < \alpha \sqrt { N}$, since ${\sqrt {\frac{\beta }{{  N}}}  < \alpha }$, it follows that 
\begin{equation}
{S_{{\text{sub}}}} = \left\{ {\left. {{\mathbf{v}} \in {\mathbb{C}^{ N}}} \right|\left| {{v_i}} \right| \leqslant \sqrt {\frac{\beta }{{  N}}} ,\forall i \in \left[ {1, N} \right]} \right\} \subseteq {S_3}
\label{eqa74}.
\end{equation}
Moreover, for any ${\mathbf{v}} \in {S_{{\text{sub}}}}$, we have 
\begin{equation}
{{\mathbf{v}}^H}{\mathbf{v}} = \sum\limits_{i = 1}^{ N} {{{\left| {{v_i}} \right|}^2}}  \leqslant   N\frac{\beta }{{ N}} = \beta,
\end{equation}
which implies ${S_{{\text{sub}}}} \subseteq {S_4}$. Combining the two results yields ${S_{{\text{sub}}}} \subseteq {S_3} \cap {S_4}$. This completes the proof.

\subsection{Proof of Theorem \ref{cor2}}	
\label{appc2}
Since any two elements in ${{\bar S}_3}$, such as ${\mathbf{x}}_n$ and ${\mathbf{x}}_m$, satisfy
\begin{equation}
\begin{gathered}
  E\left( {{\mathbf{x}}_n^H{{\mathbf{x}}_m}} \right)\mathop  = \limits^{\left( a \right)} \sum\limits_{i = 1}^N {E\left( {{\rho_{n,i}}{e^{ - j{\varphi _{n,i}}}}{\rho_{m,i}}{e^{j{\varphi _{m,i}}}}} \right)}  \hfill \\
   = \sum\limits_{i = 1}^N {E\left( {{\rho_{n,i}}{\rho_{m,i}}} \right)E\left( {{e^{j{\varphi _{n,i}}}}{e^{j{\varphi _{m,i}}}}} \right)} \mathop  = \limits^{\left( b \right)} 0 \hfill \\ 
\end{gathered}
\label{eq_app4_1}
\end{equation}
where equation (a) follows from the fact that each element of ${{\mathbf{x}}_n}$ can be written as ${x_{n,i}} = {\rho_{n,i}}{e^{j{\varphi _{n,i}}}}$ with ${\rho_{n,i}} \leqslant \gamma$. Equation (b) follows from the fact that $\varphi_{n,i}$ and $\varphi_{m,i}$ are independent and uniformly distributed over $\left[ {0,2\pi } \right)$ \cite{zhu2024robust}, i.e.,
\begin{equation}
E\left( {{e^{j{\varphi _{n,i}}}}{e^{j{\varphi _{m,i}}}}} \right) = \int_0^{2\pi } {\frac{1}{{2\pi }}{e^{j{\varphi _n}}}d{\varphi _n}}  \cdot \int_0^{2\pi } {\frac{1}{{2\pi }}{e^{j{\varphi _m}}}d{\varphi _m}}  = 0.
\end{equation}
Thus, by the law of large numbers, \eqref{eq_app4_1} implies that as $N \to \infty $, any two distinct elements in ${{\bar S}_3}$ are orthogonal, i.e., the angle $\theta$ between them satisfies
\begin{equation}
\theta=\frac{\pi }{2}
\label{eq_theta_0}. 
\end{equation}
Moreover, ${{\bar S}_1}$ can be viewed as an affine subspace located at a distance $d = \left| {{\text{Pro}}{{\text{j}}_{{{\left( {{\text{null}}\left( {\mathbf{G}} \right)} \right)}^ \bot }}}\left( {{{\mathbf{x}}_0}} \right)} \right|$ from the origin, and ${{\bar S}_3}$ consists of lines contained in a complex ball of radius $\gamma \sqrt{N}$. Their relationship is illustrated in Fig.~\ref{intersect_s1_s3}, where the green lines represent ${{\bar S}_3}$ and the blue plane represents ${{\bar S}_1}$. 
\begin{figure}[!ht] 
	\centering\includegraphics[width=3.5cm]{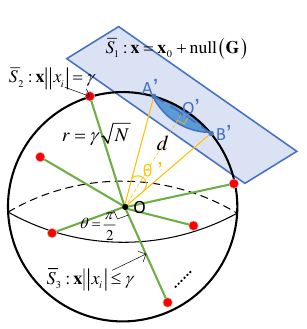}
    \caption{Relationship between ${{\bar S}_1}$ and ${{\bar S}_3}$ in high-dimensional complex space.}
	\label{intersect_s1_s3}
\end{figure}

When the complex sphere (${\mathbb{S}^r}$) intersects ${{\bar S}_1}$, we denote the resulting intersection by ${A'}{O'}{B'}$, where ${A'}$ and ${B'}$ are the endpoints of a diameter of the intersection. Thus, the central angle ${\theta '}$ is given by
\begin{equation}
\theta ' = {\text{arccos}}\frac{{\left\langle {{{\mathbf{x}}_{A'}},{{\mathbf{x}}_{B'}}} \right\rangle }}{{\left| {{{\mathbf{x}}_{A'}}} \right|\left| {{{\mathbf{x}}_{B'}}} \right|}},
\label{eq_s1_s2_theta}
\end{equation}
where ${{\mathbf{x}}_{A'}} = d{\mathbf{\hat n}} + {\mathbf{w}}$, ${{\mathbf{x}}_{B'}} = d{\mathbf{\hat n}} - {\mathbf{w}}$ with ${{\mathbf{\hat n}}}$ denoting the normal vector of ${{\bar S}_1}$, and ${\mathbf{w}} \triangleq \overrightarrow {O'A'}$ satisfying $\left| {\mathbf{w}} \right| = \sqrt {{r^2} - {d^2}}$. Thus, ${\mathbf{\hat n}} \bot {\mathbf{w}}$ and \eqref{eq_s1_s2_theta} can be rewritten as
\begin{equation}
\theta ' = {\text{arccos}}\frac{{{{\left( {d{\mathbf{\hat n}} + {\mathbf{w}}} \right)}^H}\left( {d{\mathbf{\hat n}} - {\mathbf{w}}} \right)}}{{{r^2}}} = {\text{arccos}}\frac{{2{d^2} - {r^2}}}{{{r^2}}}.
\end{equation}

Consequently, suppose, for the sake of contradiction, that when $\pi  \geqslant \theta ' \geqslant \theta  = \frac{\pi }{2}$, the intersection between ${{\bar S}_1}$ and ${{\bar S}_3}$ is empty, i.e., ${{\bar S}_1} \cap {{\bar S}_3} = \varnothing$. This would imply that all vectors in ${{\bar S}_3}$ lie outside the cone, meaning that there must exist two vectors in ${{\bar S}_3}$ whose mutual angle is strictly larger than $\frac{\pi }{2}$, which contradicts \eqref{eq_theta_0}. Therefore, when $\pi  \geqslant \theta ' \geqslant \frac{\pi }{2}$, i.e.,
\begin{equation}
\frac{{2{d^2} - {r^2}}}{{{r^2}}} \leqslant 0 \Leftrightarrow d \leqslant \frac{r}{{\sqrt 2 }} \Leftrightarrow \left| {{\text{Pro}}{{\text{j}}_{{{\left( {{\text{null}}\left( {\mathbf{G}} \right)} \right)}^ \bot }}}\left( {{{\mathbf{x}}_0}} \right)} \right| \leqslant \frac{{\gamma \sqrt N }}{{\sqrt 2 }},
\end{equation}
we must have ${{\bar S}_1} \cap {{\bar S}_3} \ne \varnothing$, which completes the proof.

\subsection{Proof of expression \eqref{eq65}}	
\label{appF}
By combining the results from Cases 1, 2, and 3, it directly follows that
\begin{equation}
N \geqslant \left\{ \begin{gathered}
  {N_{{\text{suf1}}}}{\text{, if }}\frac{\beta }{{{\alpha ^2}}} \geqslant N \hfill \\
  {N_{{\text{suf2}}}},{\text{ if }}\frac{\beta }{{{\alpha ^2}}} < N \hfill \\ 
\end{gathered}  \right.
\label{F88},
\end{equation}
where ${N_{{\text{suf1}}}} = \frac{1}{2}\left( {a + \sqrt {{a^2} + \frac{8}{{{\alpha ^2}}}b} } \right)$ and ${N_{{\text{suf2}}}} = \frac{{2b}}{\beta } + a$.
However, since the number of REs ($N$) is an unknown parameter needing to be determined, its presence in the piecewise condition makes \eqref{F88} inapplicable. 

To this end, we first consider the case where the equivalent reflection power $\beta$ is sufficiently large such that $\beta  \geqslant {N_{{\text{suf1}}}}{\alpha ^2}$. In this case, if $\beta  \geqslant N{\alpha^2}$, according to \eqref{F88}, interference-free transmission requires $N \geqslant {N_{{\text{suf1}}}}$. Thus we have 
\begin{equation}
{N_{{\text{suf1}}}} \leqslant N \leqslant \frac{\beta }{{{\alpha ^2}}}
\label{eq_D1_case1}.
\end{equation}

Furthermore, if $\beta  < N{\alpha ^2}$, then according to \eqref{F88}, it must also hold that $N \geqslant {N_{{\text{suf2}}}}$. Thus, to ensure interference-free transmission, $N$ should satisfy 
\begin{equation}
N \geqslant \max \left( {{N_{{\text{suf2}}}},\frac{\beta }{{{\alpha ^2}}}} \right).
\label{F90}
\end{equation}
For \eqref{F90}, if we assume ${N_{{\text{suf2}}}} > \frac{\beta }{{{\alpha ^2}}}$, there must be
\begin{equation}
{N_{{\text{suf2}}}} > \frac{\beta }{{{\alpha ^2}}} \Leftrightarrow \beta  < \frac{{{\alpha ^2}}}{2}\left( {a + \sqrt {{a^2} + \frac{8}{{{\alpha ^2}}}b} } \right) = {\alpha ^2}{N_{{\text{suf1}}}}
\label{F93},
\end{equation}
which contradicts the precondition of $\beta  \geqslant {N_{{\text{suf1}}}}{\alpha ^2}$. Thus, $\frac{\beta }{{{\alpha ^2}}}$ must larger than ${N_{{\text{suf2}}}}$, and \eqref{F90} can be simplified as
\begin{equation}
N \geqslant \max \left( {{N_{{\text{suf2}}}},\frac{\beta }{{{\alpha ^2}}}} \right)= \frac{\beta }{{{\alpha ^2}}}
\label{eq_D1_case2}.
\end{equation}
Therefore, by combining \eqref{eq_D1_case1}, \eqref{eq_D1_case2}, and the precondition of $\beta  \geqslant {N_{{\text{suf1}}}}{\alpha ^2}$, the sufficient condition for interference-free transmission in this case can be expressed as
\begin{equation}
N \geqslant {N_{{\text{suf1}}}}{\text{, if }}\beta  \geqslant {N_{{\text{suf1}}}}{\alpha ^2}
\label{eq_D1}.
\end{equation}

Similarly, consider the case where the equivalent reflection power $\beta$ is insufficient, i.e., $\beta < {N_{\text{suf1}}}{\alpha^2}$. If $\beta \geqslant N{\alpha^2}$, then according to \eqref{F88}, we have
\begin{equation}
{N_{{\text{suf1}}}} \leqslant N \leqslant \frac{\beta }{{{\alpha ^2}}} \Rightarrow {N_{{\text{suf1}}}}{\alpha ^2} \leqslant \beta,
\end{equation}
which contradicts the precondition $\beta  < {N_{{\text{suf1}}}}{\alpha ^2}$. Furthermore, if $\beta  < N{\alpha ^2}$, then based on the previous analysis, we have 
\begin{equation}
N \geqslant \max \left( {{N_{{\text{suf2}}}},\frac{\beta }{{{\alpha ^2}}}} \right) = {N_{{\text{suf2}}}}
\label{eq_D2_case2}.
\end{equation}
Therefore, in this case, the sufficient condition for interference-free transmission is given by
\begin{equation}
N \geqslant {N_{{\text{suf2}}}}{\text{, if }}\beta  < {N_{{\text{suf1}}}}{\alpha ^2}
\label{eq_D2}.
\end{equation}

Combining \eqref{eq_D1} and \eqref{eq_D2} yields \eqref{eq65}, completing the proof.

\subsection{Proof of the Asymptotic Form of ${N_{{\text{suf}}}}$ in \eqref{eq_ars}}
\label{appN_ars}
According to the definitions of ${N_{{\text{suf1}}}}$, ${N_{{\text{suf2}}}}$, and ${N_{{\text{suf}}}}$ in \eqref{eq60}, \eqref{eq63}, and \eqref{eq65}, respectively, it follows that 
\begin{equation}
{\text{when }}\frac{\beta }{{{\alpha ^2}}} \geqslant {N_{{\text{suf1}}}},{N_{{\text{suf}}}} = {N_{{\text{suf1}}}} = \left\{ \begin{gathered}
  O\left( {{K^2}} \right),{\text{ if }}\eta  < \alpha K \hfill \\
  O\left( {\frac{K}{\alpha }\eta } \right),{\text{ if }}\eta  > \alpha K \hfill \\ 
\end{gathered}  \right.
\label{I1},
\end{equation}
\begin{equation}
{\text{when }}\frac{\beta }{{{\alpha ^2}}} < {N_{{\text{suf1}}}},{N_{{\text{suf}}}} = {N_{{\text{suf2}}}} = \left\{ \begin{gathered}
  O\left( {{K^2}} \right),{\text{ if }}\eta  < \sqrt \beta   \hfill \\
  O\left( {\frac{{{K^2}}}{\beta }{\eta ^2}} \right),{\text{ if }}\eta  > \sqrt \beta   \hfill \\ 
\end{gathered}  \right.
\label{I2}.
\end{equation}

Therefore, if $\eta  < \alpha K$, we have ${N_{{\text{suf1}}}} = O\left( {{K^2}} \right)$, and thus
\begin{equation}
\left\{ \begin{gathered}
  {\text{when }}\frac{\beta }{{{\alpha ^2}}} > {K^2} \Leftrightarrow \sqrt \beta   > K\alpha ,{\text{ }}{N_{{\text{suf}}}} = O\left( {{K^2}} \right) \hfill \\
  {\text{when }}\sqrt \beta   < K\alpha ,{\text{ }}{N_{{\text{suf}}}} = \left\{ \begin{gathered}
  O\left( {{K^2}} \right){\text{, if }}\eta  < \sqrt \beta   \hfill \\
  O\left( {\frac{{{K^2}}}{\beta }{\eta ^2}} \right),{\text{ if }}\eta  > \sqrt \beta   \hfill \\ 
\end{gathered}  \right. \hfill \\ 
\end{gathered}  \right.
\label{eq_ars_O1}.
\end{equation}

Similarly, if $\eta  > \alpha K$, we have ${N_{{\text{suf1}}}} = O\left( {\frac{K}{\alpha }\eta } \right)$, and thus
\begin{equation}
\left\{ \begin{gathered}
  {\text{when }}\frac{\beta }{{{\alpha ^2}}} > \frac{K}{\alpha }\eta  \Leftrightarrow \frac{\beta }{{K\alpha }} > \eta ,{\text{ }}{N_{{\text{suf}}}} = O\left( {\frac{K}{\alpha }\eta } \right) \hfill \\
  {\text{when }}\frac{\beta }{{K\alpha }} < \eta ,{\text{ }}{N_{{\text{suf}}}} = O\left( {\frac{{{K^2}}}{\beta }{\eta ^2}} \right),{\text{ if }}\eta  > \sqrt \beta   \hfill \\ 
\end{gathered}  \right.
\label{eq_ars_O2}.
\end{equation}
The second term results from the fact that, when $\frac{\beta}{K\alpha} < \eta$, ${N_{\text{suf}}} = O(K^2)$ if and only if $\eta < \sqrt{\beta}$. Consequently,
\begin{equation}
\frac{\beta }{{K\alpha }} < \eta  < \sqrt \beta   \Rightarrow \sqrt \beta   < K\alpha  \Rightarrow \eta  < K\alpha,
\end{equation}
which contradicts the precondition $\eta  > \alpha K$.

Thus, by combining \eqref{eq_ars_O1} and \eqref{eq_ars_O2}, we obtain that when
\begin{equation}
\begin{aligned}
   \eta  < \sqrt \beta   < \alpha K {\text{ or }}\eta  < \alpha K < \sqrt \beta
   \Leftrightarrow \eta  < \min \left( {\alpha K,\sqrt \beta  } \right)
\end{aligned}
,
\end{equation}
${N_{{\text{suf}}}} = O\left( {{K^2}} \right)$. When $\eta  > \alpha K$ and $\eta  < \frac{\beta }{{\alpha K}}$, ${N_{{\text{suf}}}} = O\left( {\frac{K}{\alpha }\eta } \right)$. Otherwise, ${N_{{\text{suf}}}} = O\left( {\frac{{{K^2}}}{\beta }{\eta ^2}} \right)$. 

\subsection{Proof of expression \eqref{eq70}}	
\label{appG}
Since \eqref{eq69} can be rewritten as
\begin{equation}
N \geqslant \left\{ \begin{gathered}
  \frac{1}{2}\left( {a + \sqrt {{a^2} + \frac{8}{{{\alpha ^2}}}b} } \right) \triangleq {\bar N_{{\text{suf1}}}}{\text{, if cond1 holds}} \hfill \\
  \frac{{2b\mu }}{{{P_{{\text{tot}}}} - N{P_{{\text{cir}}}}}} + a \triangleq {\bar N_{{\text{suf2}}}},{\text{ otherwise}} \hfill \\ 
\end{gathered}  \right.
\label{G97},
\end{equation}
where cond1 represents ${P_{{\text{tot}}}} \geqslant  \mu {\alpha ^2}{\bar N_{{\text{suf1}}}} + N{P_{{\text{cir}}}}$, which is equivalent to $N \leqslant \frac{{{P_{{\text{tot}}}} - \mu {\alpha ^2}{\bar N_{{\text{suf1}}}}}}{{{P_{{\text{cir}}}}}}$.\\
\textbf{Scenario 1:} \textit{sufficiently large total power ${{P_{{\text{tot}}}}}$}

We first consider the scenario where the given total power budget ${{P_{{\text{tot}}}}}$ is sufficiently large such that
\begin{equation}
{P_{{\text{tot}}}} \geqslant \left( {{P_{{\text{cir}}}} + \mu {\alpha ^2}} \right){\bar N_{{\text{suf1}}}}
\label{G98}.
\end{equation}

Consequently, according to \eqref{G97}, if $N \leqslant \frac{{{P_{{\text{tot}}}} - \mu {\alpha ^2}{\bar N_{{\text{suf1}}}}}}{{{P_{{\text{cir}}}}}}$, it must also hold that $N \geqslant {\bar N_{{\text{suf1}}}}$. Thus, we have
\begin{equation}
{\bar N_{{\text{suf1}}}} \leqslant N \leqslant \frac{{{P_{{\text{tot}}}} - \mu {\alpha ^2}{\bar N_{{\text{suf1}}}}}}{{{P_{{\text{cir}}}}}}\triangleq {N_c}
\label{G99},
\end{equation}
and the sufficient condition in this case can be expressed as
\begin{equation}
{N_c} \geqslant N \geqslant {\bar N_{{\text{suf1}}}}{\text{, if }}{P_{{\text{tot}}}} \geqslant \left( {{P_{{\text{cir}}}} + \mu {\alpha ^2}} \right){\bar N_{{\text{suf1}}}}
\label{eq_sufficient_C1}.
\end{equation}

In contrast, if $N > {N_c}$, then according to \eqref{G97}, $N$ must also satisfy
\begin{equation}
N \geqslant {\bar N_{{\text{suf2}}}} \Leftrightarrow {{\bar N}_2} \geqslant N \geqslant {{\bar N}_1}
\label{G104},
\end{equation}
where
\begin{equation}
{{\bar N}_{\text{1}}} = \frac{{{P_{{\text{tot}}}} + a{P_{{\text{cir}}}} - \sqrt {{{\left( {{P_{{\text{tot}}}} + a{P_{{\text{cir}}}}} \right)}^2} - 4{P_{{\text{cir}}}}\left( {2b\mu  + a{P_{{\text{tot}}}}} \right)} }}{{2{P_{{\text{cir}}}}}} 
\label{n1},
\end{equation}
\begin{equation}
{{\bar N}_2} = \frac{{{P_{{\text{tot}}}} + a{P_{{\text{cir}}}} + \sqrt {{{\left( {{P_{{\text{tot}}}} + a{P_{{\text{cir}}}}} \right)}^2} - 4{P_{{\text{cir}}}}\left( {2b\mu  + a{P_{{\text{tot}}}}} \right)} }}{{2{P_{{\text{cir}}}}}} 
\label{n2}.
\end{equation}
Moreover, the ${P_{{\text{tot}}}}$ should satisfy
\begin{equation}
{P_{{\text{tot}}}} \geqslant a{P_{{\text{cir}}}} + \sqrt {8{P_{{\text{cir}}}}b\mu }  \triangleq {P_\Delta }
\label{G105}.
\end{equation}
Therefore, the sufficient condition in this case is given by 
\begin{equation}
\left( {N > {N_c}} \right) \cap \left( {{{\bar N}_2} \geqslant N \geqslant {{\bar N}_1}} \right).
\label{eq_N_S21}
\end{equation}

Firstly, if we assume ${{\bar N}_1} > {N_c}$, then
\begin{equation}
\begin{aligned}
{{\bar N}_1} > {N_c} \Leftrightarrow  - &\sqrt {{{\left( {{P_{{\text{tot}}}} + a{P_{{\text{cir}}}}} \right)}^2} - 4{P_{{\text{cir}}}}\left( {2b\mu  + a{P_{{\text{tot}}}}} \right)}  \hfill \\
   &> {P_{{\text{tot}}}} - 2\mu {\alpha ^2}{\bar N_{{\text{suf1}}}} - a{P_{{\text{cir}}}} \hfill \\ 
\end{aligned}
\label{G106}.
\end{equation}
If ${P_{{\text{tot}}}}  \geqslant  2\mu {\alpha ^2}{\bar N_{{\text{suf1}}}} + a{P_{{\text{cir}}}}$, then inequality \eqref{G106} is infeasible.
If ${P_{{\text{tot}}}} < 2\mu {\alpha ^2}{\bar N_{{\text{suf1}}}} + a{P_{{\text{cir}}}}$, then the inequality \eqref{G106} can be transformed as
\begin{equation}
{P_{{\text{tot}}}} < \mu {\alpha ^2}{\bar N_{{\text{suf1}}}} + a{P_{{\text{cir}}}} + \frac{{2b{P_{{\text{cir}}}} }}{{ {\alpha ^2}{\bar N_{{\text{suf1}}}}}}
\label{G107}.
\end{equation}
Substituting ${\bar N_{{\text{suf1}}}} = \frac{1}{2}\left( {a + \sqrt {{a^2} + \frac{8}{{{\alpha ^2}}}b} } \right)$ into \eqref{G107}, we have
\begin{equation}
\mu {\alpha ^2}{\bar N_{{\text{suf1}}}} + a{P_{{\text{cir}}}} + \frac{{2b{P_{{\text{cir}}}}}}{{{\alpha ^2}{\bar N_{{\text{suf1}}}}}} = \left( {{P_{{\text{cir}}}} + \mu {\alpha ^2}} \right){\bar N_{{\text{suf1}}}}
\label{G108},
\end{equation}
it implies that \eqref{G107} contradicts the precondition \eqref{G98}. Thus, it must hold that ${{\bar N}_1} \leqslant {N_c}$.

Similarly, if we assume ${{\bar N}_2} < {N_c}$, then
\begin{equation}
\begin{aligned}
{{\bar N}_2} < {N_c} \Leftrightarrow  &\sqrt {{{\left( {{P_{{\text{tot}}}} + a{P_{{\text{cir}}}}} \right)}^2} - 4{P_{{\text{cir}}}}\left( {2b\mu  + a{P_{{\text{tot}}}}} \right)}  \hfill \\
   &< {P_{{\text{tot}}}} - 2\mu {\alpha ^2}{\bar N_{{\text{suf1}}}} - a{P_{{\text{cir}}}} \hfill \\ 
\end{aligned}
\label{G114}.
\end{equation}
If ${P_{{\text{tot}}}}  \leqslant  2\mu {\alpha ^2}{\bar N_{{\text{suf1}}}} + a{P_{{\text{cir}}}}$, then inequality \eqref{G114} is infeasible.
If ${P_{{\text{tot}}}} > 2\mu {\alpha ^2}{\bar N_{{\text{suf1}}}} + a{P_{{\text{cir}}}}$, then the inequality \eqref{G114} can be transformed as
\begin{equation}
{P_{{\text{tot}}}} < \mu {\alpha ^2}{\bar N_{{\text{suf1}}}} + a{P_{{\text{cir}}}} + \frac{{2b{P_{{\text{cir}}}} }}{{ {\alpha ^2}{\bar N_{{\text{suf1}}}}}},
\end{equation}
which contradicts \eqref{G98}. Thus, it must hold that ${{\bar N}_2} \geqslant {N_c}$.

Therefore, \eqref{eq_N_S21} can be simplified as $ {{{\bar N}_2} \geqslant N > {N_c}}$.
Moreover, ${P_{{\text{tot}}}}$ should satisfy
\begin{equation}
{P_{{\text{tot}}}} \geqslant \max \left( {\mu {\alpha ^2}{\bar N_{{\text{suf1}}}} + a{P_{{\text{cir}}}} + \frac{{2b{P_{{\text{cir}}}}}}{{{\alpha ^2}{\bar N_{{\text{suf1}}}}}},{P_\Delta }} \right)
\label{G118}.
\end{equation}
According to the AM-GM inequality, we have
\begin{equation}
\mu {\alpha ^2}{\bar N_{{\text{suf1}}}} + \frac{{2b{P_{{\text{cir}}}}}}{{{\alpha ^2}{\bar N_{{\text{suf1}}}}}} \geqslant \sqrt {8b{P_{{\text{cir}}}}\mu }.
\end{equation}
Thus,
\begin{equation}
\mu {\alpha ^2}{\bar N_{{\text{suf1}}}} + a{P_{{\text{cir}}}} + \frac{{2b{P_{{\text{cir}}}}}}{{{\alpha ^2}{\bar N_{{\text{suf1}}}}}} \geqslant a{P_{{\text{cir}}}} + \sqrt {8b{P_{{\text{cir}}}}\mu }  \triangleq {P_\Delta }
\label{G121}.
\end{equation}
Therefore, the sufficient condition in this case is given by
\begin{equation}
{{\bar N}_2} \geqslant N > {N_c}{\text{, if }}{P_{{\text{tot}}}} \geqslant \left( {{P_{{\text{cir}}}} + \mu {\alpha ^2}} \right){\bar N_{{\text{suf1}}}}
\label{eq_sufficient_C2}.
\end{equation}

Consequently, combining \eqref{eq_sufficient_C1} and \eqref{eq_sufficient_C2}, we conclude that, if the total power budget is sufficient, the sufficient condition on the required number of active REs is
\begin{equation}
{{\bar N}_2} \geqslant N \geqslant {\bar N_{{\text{suf1}}}}{\text{, if }}{P_{{\text{tot}}}} \geqslant \left( {{P_{{\text{cir}}}} + \mu {\alpha ^2}} \right){\bar N_{{\text{suf1}}}}
\label{eq_N1}.
\end{equation}
\textbf{Scenario 2:} \textit{insufficient total power ${{P_{{\text{tot}}}}}$}

Now, we consider the scenario where the given total power budget ${{P_{{\text{tot}}}}}$ is insufficient such that
\begin{equation}
{P_{{\text{tot}}}} < \left( {{P_{{\text{cir}}}} + \mu {\alpha ^2}} \right){\bar N_{{\text{suf1}}}}
\label{G123}.
\end{equation}
Consequently, according to \eqref{G97}, if $N \leqslant \frac{{{P_{{\text{tot}}}} - \mu {\alpha ^2}{\bar N_{{\text{suf1}}}}}}{{{P_{{\text{cir}}}}}}$, the sufficient condition on $N$ in this case is given by
\begin{equation}
\begin{gathered}
  {\bar N_{{\text{suf1}}}} \leqslant N \leqslant \frac{{{P_{{\text{tot}}}} - \mu {\alpha ^2}{\bar N_{{\text{suf1}}}}}}{{{P_{{\text{cir}}}}}} \triangleq {N_c} \hfill \\
   \Leftrightarrow {P_{{\text{tot}}}} \geqslant \left( {{P_{{\text{cir}}}} + \mu {\alpha ^2}} \right){\bar N_{{\text{suf1}}}} \hfill \\ 
\end{gathered}.
\end{equation}
It contradicts the precondition \eqref{G123}.

In contrast, for $N > {N_c}$, according to \eqref{G97}, \eqref{G104} and \eqref{G105}, we have $N \geqslant {{\bar N}_{{\text{suf2}}}} \Leftrightarrow {{\bar N}_2} \geqslant N \geqslant {{\bar N}_1} \text{ and } {P_{{\text{tot}}}} \geqslant {P_\Delta }$.
Thus, the sufficient condition can be expressed as
\begin{equation}
\left( {N > {N_c}} \right) \cap \left( {{\bar N}_2} \geqslant N \geqslant {{\bar N}_1} \right)
\label{eq_s1}.
\end{equation}

According to \eqref{G114}, if ${P_{{\text{tot}}}} \leqslant 2\mu {\alpha ^2}{\bar N_{{\text{suf1}}}} + a{P_{{\text{cir}}}} \triangleq {P_1}$, then ${N_c} \leqslant {{\bar N}_2}$.
If ${P_{{\text{tot}}}} > {P_1}$ and ${P_{{\text{tot}}}} < \left( {{P_{{\text{cir}}}} + \mu {\alpha ^2}} \right){\bar N_{{\text{suf1}}}} \triangleq {P_2}$, then ${N_c} > {{\bar N}_2}$. It follows that
\begin{equation}
\begin{gathered}
  {\text{when }}{P_1} < {P_2},\left\{ \begin{gathered}
  {\text{if }}{P_{{\text{tot}}}} \leqslant {P_1} < {P_2}{\text{, then }}{N_c} \leqslant {{\bar N}_2}, \hfill \\
  {\text{if }}{P_1} < {P_{{\text{tot}}}} < {P_2}{\text{, then }}{N_c} > {{\bar N}_2}, \hfill \\ 
\end{gathered}  \right. \hfill \\
  {\text{when }}{P_1}  \geqslant {P_2},{P_{{\text{tot}}}} < {P_2} \leqslant {P_1}{\text{, then }}{N_c} < {{\bar N}_2}. \hfill \\ 
\end{gathered}
\label{eq_s2}
\end{equation}

According to \eqref{G106}, if ${P_{{\text{tot}}}} \geqslant {P_1}$, then ${{\bar N}_1} \leqslant {N_c}$. If ${P_{{\text{tot}}}} < {P_1}$ and ${P_{{\text{tot}}}} < {P_2}$, then ${{\bar N}_1} > {N_c}$. It follows that
\begin{equation}
\begin{gathered}
  {\text{when }}{P_1} < {P_2},\left\{ \begin{gathered}
  {\text{if }}{P_{{\text{tot}}}} \leqslant {P_1} < {P_2}{\text{, then }}{N_c} \leqslant {{\bar N}_1}, \hfill \\
  {\text{if }}{P_1} < {P_{{\text{tot}}}} < {P_2}{\text{, then }}{N_c} > {{\bar N}_1}, \hfill \\ 
\end{gathered}  \right. \hfill \\
  {\text{when }}{P_1} \geqslant {P_2},{P_{{\text{tot}}}} < {P_2} \leqslant {P_1}{\text{, then }}{N_c} < {{\bar N}_1}. \hfill \\ 
\end{gathered}
\label{eq_s3}
\end{equation}
Combining \eqref{eq_s1}, \eqref{eq_s2} and \eqref{eq_s3}, we have
\begin{equation}
\begin{gathered}
  \left( {N > {N_c}} \right) \cap \left( {{{\bar N}_2} \geqslant N \geqslant {{\bar N}_1}} \right) \hfill \\
   = \left\{ \begin{gathered}
  \varnothing {\text{, if }}{P_1} < {P_{{\text{tot}}}} < {P_2}, \hfill \\
  {{\bar N}_2} \geqslant N \geqslant {{\bar N}_1}{\text{, if }}{P_{{\text{tot}}}} \leqslant \min \left( {{P_1},{P_2}} \right). \hfill \\ 
\end{gathered}  \right. \hfill \\ 
\end{gathered}
\end{equation}
Moreover, for the second term, since ${P_{{\text{tot}}}} \geqslant {P_\Delta }=a{P_{{\text{cir}}}} + \sqrt {8{P_{{\text{cir}}}}b\mu }$ and ${P_{{\text{tot}}}} \leqslant P_1 = 2\mu {\alpha ^2}{\bar N_{{\text{suf1}}}} + a{P_{{\text{cir}}}}$, it follows that
\begin{equation}
\begin{gathered}
  \sqrt {8{P_{{\text{cir}}}}b\mu }  \leqslant {P_{{\text{tot}}}} - a{P_{{\text{cir}}}} \leqslant 2\mu {\alpha ^2}{\bar N_{{\text{suf1}}}} \hfill\\
 \Rightarrow  \frac{{2{P_{{\text{cir}}}}b}}{{{\alpha ^2}{\bar N_{{\text{suf1}}}}}} \leqslant \mu {\alpha ^2}{\bar N_{{\text{suf1}}}} \hfill \\ 
\end{gathered}
\label{eq_P1P2},
\end{equation}
based on \eqref{G108}, \eqref{G121} and \eqref{eq_P1P2}, we have
\begin{equation}
\begin{gathered}
  \mu {\alpha ^2}{{\bar N}_{{\text{suf1}}}} + a{P_{{\text{cir}}}} + \frac{{2b{P_{{\text{cir}}}}}}{{{\alpha ^2}{{\bar N}_{{\text{suf1}}}}}} \leqslant 2\mu {\alpha ^2}{{\bar N}_{{\text{suf1}}}} + a{P_{{\text{cir}}}} \hfill \\
   \Leftrightarrow {P_\Delta } \leqslant {P_2} \leqslant {P_1} \hfill \\ 
\end{gathered}
\label{G127}.
\end{equation}
Therefore, in this insufficient total power scenario, the sufficient condition on the required number of active REs is
\begin{equation}
{{\bar N}_2} \geqslant N \geqslant {{\bar N}_1}{\text{, if }}{P_\Delta } \leqslant {P_{{\text{tot}}}} < \left( {{P_{{\text{cir}}}} + \mu {\alpha ^2}} \right){\bar N_{{\text{suf1}}}}
\label{eq_N2}.
\end{equation}
Combining \eqref{eq_N1} and \eqref{eq_N2} yields \eqref{eq70}, completing the proof.

\bibliographystyle{IEEEtran} 
\bibliography{ref}

\end{document}